\documentclass[preprint,onefignum,onetabnum]{siamart251216}

\usepackage{array}

\usepackage{lipsum}
\usepackage{amsfonts}
\usepackage{algorithm}
\usepackage{algpseudocode}

\algrenewcommand\algorithmicrequire{\textbf{Input:}}
\algrenewcommand\algorithmicensure{\textbf{Output:}}

\usepackage{hyperref}
\usepackage{graphicx}
\usepackage{amsmath, amssymb}
\usepackage{xcolor}
\usepackage{cite}
\usepackage{subfig}
\usepackage{float}

\newcommand{\M}{M} 
\newcommand{\cb}{\mathcal{B}} 
\newcommand{\R}{\mathbb{R}}
\newcommand{\K}{\mathcal{K}}
\newcommand{\V}{\mathbf{v}}
\newcommand{\X}{\mathcal{X}}
\newcommand{\poset}{\mathcal{P}}

\numberwithin{equation}{section}
\ifpdf

\DeclareGraphicsExtensions{.eps,.pdf,.png,.jpg}
\else
  \DeclareGraphicsExtensions{.eps}
\fi
\crefname{hypothesis}{Hypothesis}{Hypotheses}

\newcommand{\eqfnnum}{7}
\makeatletter
\newcommand{\eqmark}{\hbox{$\m@th^{\@fnsymbol{\eqfnnum}}$}}
\newcommand{\eqfootnote}{%
  \begingroup
  \def\@fnsymbol##1{\ensuremath{\ifcase##1\or *\or \dagger\or \ddagger\or
    \mathsection\or \mathparagraph\or \|\or \#\or \dagger\dagger
    \or \ddagger\ddagger\else\@ctrerr\fi}}%
  \renewcommand{\thefootnote}{\fnsymbol{footnote}}%
  \def\@makefnmark{\hbox to\z@{$\m@th^{\@thefnmark}$\hss}}%
  \long\def\@makefntext##1{\parindent 1em\noindent
          \hbox to1.8em{\hss$\m@th^{\@thefnmark}$}##1}%
  \footnotetext[\eqfnnum]{These authors contributed equally.}%
  \endgroup}
\makeatother

\headers{Using PH to Analyze Access to Heterogeneous-Quality Resources}{S. Tymochko, G. Grindstaff,  A. Hickok, J. Luo, M. A. Porter}

\title{Using Persistent Homology to Analyze Access to Heterogeneous-Quality Resources and Heterogeneous-Severity Nuisances
\thanks{{Funding:} GG was supported by NSF MSPRF 2202895 and EPSRC Centre to Centre Research Collaboration grant EP/Z531224/1. AH was supported by NSF grant DMS-2303402. 
JL was supported by NSF grant numbers 1829071 and 2136090.
}}

\author{Sarah Tymochko\thanks{Department of Mathematics and Computer Science, College of the Holy Cross, Worcester, MA, USA
  (\email{stymochko@holycross.edu}).} \and Gillian Grindstaff\eqmark\thanks{Mathematical Institute, Oxford University, Oxford, UK
  (\email{grindstaff@maths.ox.ac.uk}).} \and Abigail Hickok\eqmark\thanks{Department of Mathematics, Columbia University, New York, NY, USA
  (\email{ah3966@columbia.edu}).} \and Jiajie Luo\eqmark\thanks{Knowledge Lab, University of Chicago, Chicago, IL, USA
  (\email{jerryluo8@uchicago.edu}).} \and Mason A. Porter\thanks{Nordita, KTH Royal Institute of Technology, Stockholm University, and Uppsala University, Stockholm, Sweden; Department of Mathematics, Stockholm University, Stockholm, Sweden; Santa Fe Institute, Santa Fe, NM, USA; Department of Mathematics, University of California, Los Angeles, CA, USA; Department of Sociology, University of California, Los Angeles, CA, USA (\email{mason.porter@su.se}).}}

\usepackage{amsopn}

\makeatletter
\newcommand*{\addFileDependency}[1]{
  \typeout{(#1)}
  \@addtofilelist{#1}
  \IfFileExists{#1}{}{\typeout{No file #1.}}
}
\makeatother

\ifpdf
\hypersetup{
  pdftitle={Using PH to Analyze Access to Heterogeneous-Quality Resources},
  pdfauthor={S. Tymochko, G. Grindstaff, A. Hickok, J. Luo, M. A. Porter}
}
\fi

\begin{document}

\maketitle
\eqfootnote

\begin{abstract}
We develop a framework to use multiparameter persistent homology (PH) to examine access to heterogeneous-quality resources and exposure to heterogeneous-severity nuisances in a geographic region. 
Persistent homology, which is a type of topological data analysis {(TDA)}, has been employed previously to examine resource coverage. Unlike prior approaches, which used {one-parameter} PH to study resource coverage and nuisance exposure, our method accounts for
heterogeneous-quality resources. 
{Our framework, which employs a computationally-efficient approximation of multiparameter PH,
allows one} to study access to any resource ({or} exposure of any nuisance) using any notion of quality (or severity).
Using the city of Chicago as an example region, we employ our framework to detect clusters of poor access to public parks, overexposure to landfills, and both underexposure and overexposure to pubs and bars. 
\end{abstract}





\section{Introduction}\label{intro}

{Evaluating} equity in the access to resources --- such as public parks, hospitals, and polling sites --- is a complex problem that requires the quantification of {two key} characteristics:
\begin{itemize}
    \item[(1)]{{\em geographic access}, which depends on the locations of the resource sites, their relative proximity, and the ease of accessing them; and}
    \item[(2)]{{\em resource supply}, which depends on the quantity and/or quality of the available resources \cite{beyond_proximity} at each resource site.}
\end{itemize} 
Although people in a community often 
{feel} disparities when they seek to use
shared resources, it can be difficult to quantify characteristics (1) and (2), which also
depend on what resource one is considering. For example, geographic access typically depends on {both transportation mode and local infrastructure}, resource supplies are limited, and notions of equity include subjective ideas.

In this paper, we employ topological data analysis (TDA)~\cite{dey2022} and present an exploratory approach
to examine the access of different regions in a city to public resources and thereby identify regions with poor access. We also use our approach to examine regions with overexposure to nuisances (such as pollutants and hazards). Regions can have poor access to a resource either due to their geographic distance from resource sites or to the poor quality of nearby resource sites. {Analogously, regions can experience overexposure to nuisances due to both geographic distance and nuisance severity.}
We take inspiration from multiparameter persistent homology (PH)~\cite{Carlsson2009, Botnan2022}, {which is} a generalization of {one-parameter PH~\cite{otter2017}} that allows one to analyze {multi}-scale topological structures 
that depend simultaneously on two or more parameters.
In particular, two-parameter PH {tracks} how the topological structure of a space changes through a \textit{bifiltration}, {which is} a nested sequence of topological spaces {that is} parameterized by two values.
Multiparameter PH has well-documented computational and interpretive challenges~\cite{Carlsson2009, Lesnick2015, clause_meta-diagrams_2023, otter2019}, and many researchers have thus developed workarounds to use multiparameter PH in {practice. See Section~\ref{sssec:mph_relatedwork} for a discussion of these workarounds.}

{
By using a bifiltration, our PH-inspired approach incorporates both distance and quality to assess resource accessibility and nuisance severity.} 
To do this, we build a data structure that is specific to the shape of a geographic region (e.g., the city of Chicago) and vary both distance and quality (which are our two parameters).
Calculating PH then reveals how the identified poor-access regions and nuisance overexposure vary with respect to these parameters. 
{Throughout our paper, we refer to desirable sites as ``resources'' and to undesirable sites as ``nuisances''. Both resources and nuisances have associated equity and proximity concerns.
We illustrate our framework by examining access and exposure to public parks, public houses (i.e., pubs) and bars\footnote{{In the} present paper, when we refer to {``pubs", we mean both} pubs and bars.}, and landfills.}


\subsection{Our Contributions} 

In the present paper, we {use} TDA-inspired methodology to study {the accessibility of geographically-distributed resources with heterogeneous qualities and the exposure to nuisances with heterogeneous severities.}
We use ideas from multiparameter PH to simultaneously examine both the quality and access of resource sites (and both the severity and exposure to nuisance sites). 
Our approach provides a computationally efficient framework {to approximate} two-parameter {PH}.
{This approach, which we illustrate on three case studies, is applicable to any resource (respectively, nuisance) with any scaler notion of quality (respectively, severity).}


\subsubsection{Mathematical Contributions}

{Our} framework pragmatically leverages {PH-inspired} insights to analyze joint spatial patterns that depend on multiple parameters.
For each resource or nuisance, we construct a bifiltration (see Section~\ref{sec:our_bifiltration}) 
of its associated geographic region (e.g., the city of Chicago) that incorporates both (1) the distance to the resource or nuisance site and (2) the quality/severity of the site. 
The connected components of these bifiltrations indicate 
disadvantaged {regions} at different distance and quality/severity scales.
For resources, our bifiltration identifies {regions} that have insufficient coverage. 
The bifiltration at distance and quality parameters $(r, q)$ consists of the areas where the total resource quality within a radius $r$ is less than $q$.
Analogously, for nuisances, our bifiltration identifies {overexposed regions.} 
{The bifiltration at distance and quality parameters $(r, q)$ consists of areas where the total nuisance exposure within radius a $r$ is at least $q$.}
We obtain the undesirable {regions} by aggregating components that correspond to 0D PH classes in our bifiltration (see Section~\ref{sec:aggregation}).
This yields a visual representation of the geographic locations of the undesirable {regions}.


\subsubsection{Case Studies}

{We present three case studies {in which} we identify disadvantaged geographic {regions.}
These {regions} have either poor access to quality resources or overexposure to nuisances.}
In our first case study, we use our PH-based approach to quantify the accessibility of public parks in Chicago.
We identify regions that have poor
access to quality parks relative to their surroundings. 
We assume that it is desirable to be close
to public parks, and prior studies have demonstrated that close proximity to parks has positive effects on both physical health \cite{kaczynski_proximity_park_2014} and mental health \cite{sturm_proximity_park_2014,SCHERTZ2025102624}.
{Other researchers have demonstrated that green spaces at schools impact both the health and well-being of children and the property values in surrounding neighborhoods \cite{gorjian2025}.}

In our second case study, 
we use our PH-based approach to {quantify} the exposure to landfills in Chicago.
We identify regions that are overexposed to landfills relative to their {surrounding regions}.
We assume that it is undesirable to be near a landfill, and studies have demonstrated that proximity to landfills has adverse health affects \cite{njoku_landfill_2019,vrijheid_landfill_2000}. 

In our third case study, we use our PH-based approach to {examine both} access and exposure to
pubs in Chicago. {Because} pubs have both positive effects \cite{third_place} and negative effects 
\cite{seid_alcohol_2018,auchincloss_alcohol_2022,halonen_alcohol_living_2013}
on people in nearby {regions, we analyze them as both resources and nuisances.}
For instance, proximity to pubs provides access to social interactions while simultaneously subjecting people
to potentially unwanted exposure to irritants like noise and alcoholism.



\subsection{Related Work}

\subsubsection{{Topological Data Analysis for Resource Coverage}}

Researchers have employed TDA to study resource coverage in several situations using different notions of ``distance.'' Hickok et al.~\cite{Hickok_Jarman_Johnson_Luo_Porter_2024} used PH to analyze the coverage of polling sites in several {United States} cities. To calculate distances for their computation of PH, they considered both travel time to (and from) polling sites and waiting times at polling sites.
{Using a similar approach,} Gonzalez-Cruz et al. \cite{cruz-gonzalez2025sexual_healthcare} {computed PH with travel time as distance} to assess the accessibility of Planned Parenthood clinics and Federally Qualified Health Centers in California{, and they thereby identified} gaps in healthcare coverage and {regions} that are vulnerable to the defunding of {such clinics.}
O'Neil and Tymochko~\cite{oneil_evaluating_2024} {also employed a related PH approach, although their distance function was geographic Euclidean distance,} to study the coverage of cooling centers in four US cities{.}

{In most of the papers above, researchers used travel time as the distance for computing PH. In contrast, our present paper does not use travel time.}
When {we examine nuisance exposure}, we {use} Euclidean distance {because} factors like pollution and noise {worsen with}
proximity.
When we {examine resource coverage}, we {use} $L_1$ distance as a proxy for walking time {because} residents of a city {ideally can access} resources by walking.

\subsubsection{Topological Data Analysis {of} {Other} Geospatial and Social Systems}

In the present paper, we build on a growing body of research on applying PH and other TDA methods to geospatial systems, geographic {data}, and social systems. 
See \cite{feng2025} for a general discussion of using TDA to study social systems.

Researchers have also {used PH in a variety of geospatial applications beyond resource coverage.}
Feng and Porter {examined} voting data~\cite{feng2021} and city-street networks and other spatial data~\cite{feng2020, tda_spatial}.
Duchin et al.~\cite{duchin_homological_2020} used PH to study electoral redistricting in the {US}.
Fairchild and Lemoine \cite{fairchild2025mortality} used PH and MAPPER \cite{singh2007topological} to analyze {geospatial} mortality patterns in the US.
{Anderson et al.~\cite{anderson2022}, Kauba and Weighill \cite{kauba_demographics_2024}, and Rodr\'iguez V\'azquez and Cuerno \cite{rodriguez2025quantifying} used TDA to study various socioeconomic and demographic spatial patterns.}
{Cocoran} and Jones \cite{corcoran_topological_2023} used PH to examine several geospatial point-cloud data sets, including for an analysis of the distribution
of public houses (i.e., pubs) across many cities in the United Kingdom. Our {decision} to include an example with pubs draws inspiration from them.

Researchers have also used PH
to examine {phenomena such as street-network connectivity~\cite{corcoran_persistent_2022,feng2020,feng2021}, traffic dynamics~\cite{carmody_topological_2021}, and racial segregation~\cite{friesen_understanding_2024} in urban settings.}
Hickok et al.~\cite{covid} 
{used PH to detect spatiotemporal anomalies in COVID-19 outbreaks in
 Los Angeles and to examine
 COVID-19 vaccination uptake across} 
 New York City at a single point in time.
Arulandu et al.~\cite{arulandu2025grapevine} introduced a vineyard-based distance to measure topological dissimilarity between data sets. 
In one of their case studies, they used their 
{notion of} distance to detect the geographic separation between Black and Hispanic voters in 100 different US cities. 


\subsubsection{Multiparameter Persistent Homology} 
\label{sssec:mph_relatedwork}

As we discussed previously, multiparameter PH poses many computational and interpretive challenges~\cite{Botnan2022}.
Unlike one-parameter persistence modules, a multiparameter persistence module does not decompose into a collection of intervals, so there is no analogue of a barcode and no complete discrete invariant~\cite{Carlsson2009,Botnan2022}.
Consequently, researchers have developed a variety of workarounds that trade completeness for computability, stability, or interpretability~\cite{clause_meta-diagrams_2023,vipond_landscapes_2020,Loiseaux2023,loiseaux2025multi,Bjerkevik_2025,Lesnick2015,Lesnick2022,Scoccola2023}.
These approaches produce a \emph{summary} of a multiparameter persistence module; such summaries are well-suited for comparing data sets with each other or for classifying them.
Researchers have applied these tools to spatially-structured data \cite{vipond_multiparameter_2021,Chung_Day_Hu_2022}.

Researchers have also developed software implementations of these ideas.
RIVET \cite{rivet,Lesnick2015,Lesnick2022} provides interactive visualization of two-parameter persistence modules; a user selects a line in the parameter plane
and RIVET displays the barcode of the corresponding one-parameter slice.
{\sc Persistable}~\cite{Scoccola2023,Rolle2020} uses the degree--Rips bifiltration to perform density-based clustering that is stable with respect to the choice of density and scale parameters.
The {\tt multipers} package~\cite{multipers,Loiseaux2023, loiseaux2025multi} implements several multiparameter vectorizations for use in machine learning.

While these approaches provide helpful and robust summaries of multiparameter persistence, they do not address \emph{where} in the data a given topological feature lies.
Therefore we do not directly compute multiparameter PH or any of the aformentioned summaries in the present paper. 
Instead we develop an approach to identify the geographic regions that correspond to the topological features from the bifiltration.
In particular, in Section~\ref{sec:aggregation}, we describe how we aggregate the connected components of our bifiltration across both parameters in order to obtain geographical representations of PH classes.


\subsubsection{Geography and Urban Planning}

Unsurprisingly, one can use tools from geography and urban planning to study resource accessibility. {One influential framework is the notion of a 15-minute city, in which residents can reach daily necessities within a 15-minute walk or bicycle ride. Other approaches assess accessibility by measuring whether resources fall within a fixed radius or travel-time threshold from a given location, using methods such as the two-step floating-catchment area model. Other researchers flip the perspective and ask how far people must travel to reach a desired resource, or what proportion of residents live within an acceptable distance of one.}

One urban-planning concept that is relevant to highlight is the notion of a
``15-minute city'' {\cite{moreno2021introducing}, in which residents can access daily necessities (such as work, schools, shopping, parks, and healthcare) within a 15-minute walk or bicycle ride from their home. 
Researchers have worked to quantify how well existing cities follow a 15-minute-city framework and have discussed the inequities due to some cities {fitting into the framework better} than others \cite{bruno2024}. 
{Wu and Zhang~\cite{wu2026} used a TDA approach to map how residents of Shenzhen move through the city at a 15-minute-walk scale, characterizing both individual mobility patterns and broader neighborhood-level activity zones. In these approaches, researchers considered resources that lay within a fixed travel-time threshold of 15 minutes.}

{There are other} approaches {that} have considered regions that lie within a fixed radius from their nearest resource site \cite{radke_spatial_2000,luo_measures_2003, SCHIPPERIJN2017253}.
The two-step floating-catchment area (2SFCA) method \cite{radke_spatial_2000,luo_measures_2003} {is a} gravity-like model \cite{barbosa2018} of spatial interactions that was used initially to determine accessibility to primary-care physicians. 
2SFCA accounts for distances to physicians and the number of patients within a specified distance from a physician.
Researchers have used the 2SFCA method and extensions of it to assess
the accessibility of various resources (including elderly care \cite{liu_2sfca_elderly_care}, food outlets \cite{jin_huff-based_2sfca}, and public parks) for {several} cities in 
China \cite{dai_measuring_2022, qin_spatial_2020, xing_environmental_2020}.
{Another approach for examining the accessibility of resources is to calculate the median number of resource sites within a threshold distance~\cite{SCHIPPERIJN2017253}.}

An alternative approach is to examine how far people need to travel to access particular resources.
For example, Schindler et al.~\cite{schindler2022} studied how far residents of three European cities need to travel to access green space, and they concluded that the distances to access the most popular green spaces exceed recommended distances of accessibility policies.
In contrast to the perspective of a 15-minute city, {they} asked people to {specify their} desired green-space quality, and they then determined the distance that people needed to travel to access sufficient resources.
Other researchers have quantified resource access by examining the median distance to resource centers \cite{SCHIPPERIJN2017253,ZHANG2023104191} or the proportion of people that live within a threshold distance of a resource \cite{KaufmannVisputeKansal}.


\subsubsection{Approaches from Network Analysis and Machine Learning}

Researchers have also {employed} network-based approaches to examine resource accessibility.
Moi et al.~\cite{moi_urban_topology_2024} used a network-theoretic framework to quantify green-space accessibility in three cities in Italy.
They incorporated distance and overcrowding by combining a Voronoi tessellation with agent-based dynamics.
Boncinelli et al.~\cite{florence_walk} took inspiration from 
{the notion of a 15-minute city}  
and used network analysis to examine service accessibility in Florence, Italy. 
{They represented services by the nodes of a network,}
and they used edge weights to encode the extent of
shared access between these services. 
They identified clusters that consisted of services with shared accessibility and the neighborhood that they served.
Horton et al.~\cite{horton2025,horton2025b} used an optimization approach to examine the accessibility of groceries in walkable cities and of polling sites in cities, and Skipper et al.~\cite{skipper2025} investigated how to optimize fair geographic access to polling sites.
Weng et al.~\cite{weng2025beyonddistance} trained a neural-embedding model on observed human mobility trajectories from about 25 million personal devices 
to infer functional distances between locations in 11 cities. 
Their study revealed both visible and invisible mobility barriers in urban spaces. 
Other researchers have focused on accessibility of resources via public transit~\cite{sfil2025}.


\subsection{Organization of the Paper}

Our paper proceeds as follows. In Section~\ref{sec:background}, we give background about {PH}. In Section~\ref{sec:our_bifiltration}, we introduce our construction of a bifiltration in distance and quality. Further, we detail our multiparameter-persistence-inspired approach to identifying the geographical regions that underlie the topological features. 
In Section~\ref{sec:results}, we present the results of our {PH-inspired method in} three case studies. In Section~\ref{conc}, we conclude and discuss implications and limitations of our work.
In Appendix \ref{app:data}, we give details about data access and our data processing. 


\subsection{Data and Code}

We employ data 
from the Chicago Data Portal \cite{chicagodata}, the Landfill Methane Outreach Program (LMOP) of the United States Environmental Protection Agency~\cite{landfilldata}, OpenStreetMap (OSM) \cite{openstreetmap}, and Google Reviews \cite{GoogleAPI}.
In Appendix \ref{app:data}, we give details about data access and data processing. Our code is available at \hyperlink{https://github.com/sarahtymochko/multiparam-resource-coverage}{https://github.com/sarahtymochko/multiparam-resource-coverage} \cite{ourgithub}{.}


\section{Background about Persistent Homology}
\label{sec:background}

In this section, we briefly review relevant mathematical background from TDA and PH. For thorough introductions, see \cite{Edelsbrunner_Harer_2009,otter2017} for one-parameter PH and \cite{Botnan2022} for multiparameter PH. See \cite{dey2022} for 
an extensive introduction to TDA and PH.


\subsection{One-Parameter PH}  \label{single}
A \emph{one-parameter filtration} is a nested sequence \linebreak $\{\K_r\}_{r \in \R}$ of topological spaces that is indexed by a \emph{filtration parameter} $r$ and satisfies {either 
$\K_s \subseteq \K_r$ for all $s \leq r$ or $\K_r \subseteq \K_s$ for all $s \leq r$.}
{Equivalently, a one-parameter filtration is a either a functor $\mathcal{F}$ from $\R$ (viewed as a poset category) to ${\bf Top}$ or a functor $\mathcal{F}$ from $\R^{\mathrm{op}}$ to ${\bf Top}$, where ${\bf Top}$ is the category whose objects are topological spaces and whose morphisms are inclusion maps, and where one obtains $\R^{\mathrm{op}}$ 
by reversing the order on $\R$. A filtration where $\K_s \subseteq \K_r$ for all $s \leq r$ is equivalent to a functor $\mathcal{F}: \R \to {\bf Top}$. A filtration where $\K_r \subseteq \K_s$ for all $s \leq r$ is equivalent to a functor $\mathcal{F}: \R^{\mathrm{op}} \to {\bf Top}$.}

Given a filtration $\{\K_r\}_{r \in \R}$, one can compute the homology of each topological space $\mathcal{K}_{r}$.
A homology class represents a hole in 
$\mathcal{K}_{r}$. In the present paper, we only consider {0D homology classes}, which represent connected components. \emph{Persistent homology} (PH) tracks homology classes through a filtration by recording ``births'' 
of new homology classes and ``deaths'' of existing homology classes. 
The birth value of a homology class is the smallest
value $b$ such that the class appears in $\mathcal{K}_b$. For 0D PH, birth entails that a new component arises
in a filtration. The death value of a homology class is the smallest
value $d$ such that the class is no longer
in a filtration. For 0D PH, death entails that a 
component has merged with an existing component that was born earlier (i.e., a component with smaller birth value).
One can summarize the PH of a filtration as a multiset of the birth and death values, and one can visualize it as a scatter plot (i.e., a so-called \emph{persistence diagram}) of {death values versus birth values}.


\subsection{Naive Example of a One-Parameter Filtration for Resource and Nuisance Applications}\label{sec:single_baseline}

In our case studies, which we discuss in detail in Section~\ref{sec:results}, our data includes a geographic domain $\mathcal{D}$ (e.g., the city of Chicago), a discretization $\mathcal{B} = \{B_j\}$ of census units (which are either census blocks or census tracts) $B_j \subseteq \mathcal{D}$, and a set $\{M_i\}$ of resource sites or nuisance sites $M_i \subseteq \mathcal{D}$. For example, in Section~\ref{sec:parks}, $\{M_i\}$ is the set of public parks in Chicago.

We define a naive one-parameter filtration that is based on the distance to the closest resource site or nuisance site $M_i$.
For a census unit $B_j$, we calculate
\begin{equation*}
	g(B_j) := \min_i \, d(M_i,B_j) \,.
\end{equation*}	
Depending on the application, we compute either a superlevel-set filtration
\begin{equation*}
	g^{-1}(r_k,\infty) \subseteq g^{-1}(r_{k-1},\infty) \subseteq \cdots \subseteq g^{-1}(r_1,\infty) 
\end{equation*}
or a sublevel-set filtration
\begin{equation*}
	g^{-1}(-\infty, r_1) \subseteq g^{-1}(-\infty, r_2) \subseteq \cdots \subseteq g^{-1}(-\infty,r_k)  \,,
\end{equation*}
where $r_1 < r_2 < \cdots < r_k$. 
For resources, the superlevel-set filtration identifies the regions whose population needs to travel the farthest to access a resource site.
For nuisances, the sublevel-set filtration identifies the regions {whose populations are} most exposed to {a nuisance site}.

Let's use a superlevel-set filtration to analyze access to public parks in Chicago, which we consider in detail using a different approach in Section~\ref{sec:parks}.
In Figure~\ref{fig:Parks_Superlevel_Filtration}, we show a few steps in this filtration. 
The census units (which, in this case, are census blocks) that appear earliest in the filtration are the ones that are farthest from a park. The census units merge quickly into a single connected component, so it is difficult to draw conclusions from this naive superlevel-set filtration.

\begin{figure}
    \centering
\includegraphics[width=\linewidth]{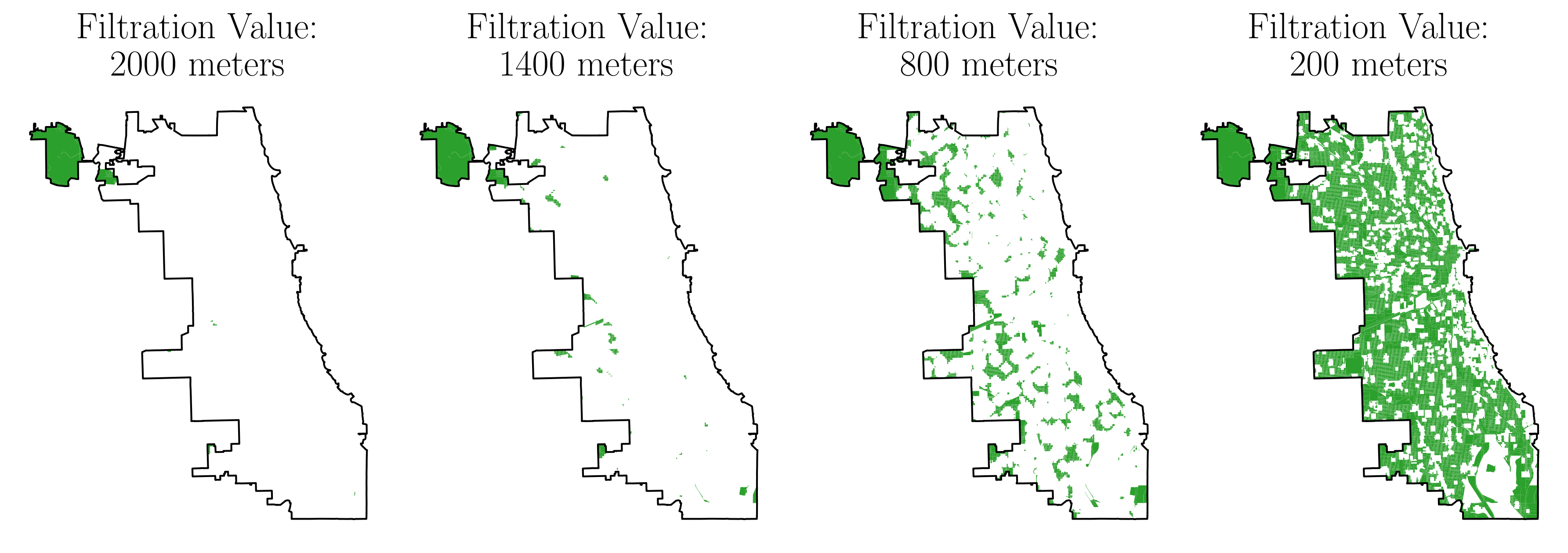}
    \caption{Four steps in a naive superlevel-set filtration for public parks in Chicago.
    } \label{fig:Parks_Superlevel_Filtration}
\end{figure}

These sublevelet-set filtrations and superlevel-set filtrations suffer from the fact that they do not account for any information about the individual sites $M_i$. 
In our case studies, the resource and nuisance sites have rather heterogeneous qualities (for resources) or severities (for nuisances). 
In Section~\ref{sec:multi}, we define a \emph{bifiltration}, which will allow us to
better assess resource coverage and nuisance severity (see Section~\ref{sec:our_bifiltration}).


\subsection{Multiparameter Persistent Homology}\label{sec:multi}
A \emph{multifiltration} is a set \linebreak $\{\K_{\V}\}_{{\V}\in \poset}$ of 
topological spaces that is indexed by a poset (i.e., partially ordered set) $\poset$ that satisfies $\K_{\mathbf{u}}\subseteq \K_{\V}$ for ${\bf u} \leq \V$.\footnote{Equivalently, a multifiltration is a functor $\mathcal{F} : \poset \to {\bf Top}$, where 
$\poset$ is
a category whose objects are the elements of $\poset$ and there is
a unique morphism $\mathbf{u} \to \V$ if $\mathbf{u} \leq \V$.} 
For example, one can consider $\mathcal{P} = \R^n$ with the standard partial order $\mathbf{u} \leq \V$ if $u_i \leq v_i$ for all $i$. 
Sometimes, it is desirable to ``reverse the order'' of some parameters. 
For example, $\R^{\textnormal{op}} \times \R$ is the poset with the underlying set $\R^2$ and the partial order $\mathbf{u} \leq \V$ whenever $u_1 \geq v_1$ and $u_2 \leq v_2$. Another example is $\R \times \R^{\textnormal{op}}$, which is the poset with the underlying set $\R^2$ and the partial order $\mathbf{u} \leq \V$ whenever $u_1 \leq v_1$ and $u_2 \geq v_2$. 
A \emph{bifiltration} is a multifiltration that is indexed by
$\R^2$, $\R^{\textnormal{op}} \times \R$, or $\R \times \R^{\textnormal{op}}$.
In Section~\ref{sec:our_bifiltration}, we consider two bifiltrations. One of them is indexed by the poset $\mathcal{P} = \R^{\textnormal{op}} \times \R$, and the other is indexed by the poset $\mathcal{P} = \R \times \R^{\textnormal{op}}$.

Given a multifiltration $\{\K_{\V}\}_{\V \in \poset}$, one can compute the homology of each topological space $\K_{\V}$ (just as one does for one-parameter PH). 
Additionally, as in one-parameter PH, an inclusion $\iota: \K_{\bf u} \subseteq \K_{\bf v}$ induces a map $\iota_* : H(\K_{\bf u}) \to H(\K_{\bf v})$ between the homology of $\K_{\bf u}$ and the homology of $\K_{\bf v}$.


\section{Our Distance--Quality Bifiltration}\label{sec:our_bifiltration}
\label{sec:methodology}

As in Section~\ref{sec:single_baseline}, let $D$ be a geographic domain (e.g., the city of Chicago)
and let $\cb = \{B_j\}$ be the set of census units $B_j \subseteq D$, which we treat
 as closed sets in $D$. Let $\M_i \subseteq D$ be the geographic region of the $i$th site (e.g., the $i$th park in our case study of public parks). We write $d(M_i, B_j)$ to denote an appropriate distance measure (which depends on the application) between site $M_i$ and census unit $B_j$. 
 In our resource case studies, $d(M_i, B_j)$ is the $L_1$ distance between resource site $M_i$ and census unit $B_j$. The distance measure $d(M_i, B_j)$ is an estimate of the travel distance on city streets between
 census unit $B_j$ and
resource site $M_i$. 
 For our nuisance case studies, $d(M_i, B_j)$ is the Euclidean distance 
 census unit $B_j$ 
 and
nuisance site $M_i$.

We assign a score $s(\M_i) \geq 0$ to indicate the quality or severity of each site $\M_i$. We thereby quantify heterogeneity in the resource and nuisance sites. 
For example, in our case study of public parks, a higher score indicates a higher-quality park. See Section~\ref{sec:parks}, Section~\ref{sec:landfills}, and Section~\ref{sec:pubs} for more details about how we define the score $s(\M_i)$ for parks, landfills, and pubs, respectively.

We define a function $f : \cb \times \mathbb{R}_{\geq 0} \to \mathbb{R}$ to quantify the cumulative accessibility (respectively, exposure) of the resource (respectively, nuisance) sites $\{M_i\}$ for each census unit $B_j$. For each census unit $B_j$ and each filtration-parameter value $r \in \R$, we define the {cumulative} (quality or severity) score
\begin{equation}
    f(B_j, r) := \sum_{\{\M_i \, \mid \, d(\M_i, B_j) \leq r\}} s(\M_i)\,.
\end{equation}
Informally, one can think of $f(B_j, r)$ as the total resource quality or total nuisance exposure
within radius $r$ of census unit $B_j$. 
For a fixed census unit $B_j$, the function $f(B_j,\cdot)$ is monotonically increasing, so $f(B_j,r_2)\leq f(B_j,r_1)$ for $r_2\leq r_1$.
 The {cumulative} score
 $f(B_j, r)$ is large when many sites $M_i$ are sufficiently close to $B_j$ or when $B_j$ has at least a few nearby sites with very large scores (or some combination of the two situations). 
 In Figure~\ref{fig:acre_block_scores_plot}, we show
the {cumulative} quality score (i.e., total resource quality)
 $f(B_j, r)$ for our case study of parks.

\begin{figure}
\centering
\includegraphics[width = .95\textwidth]{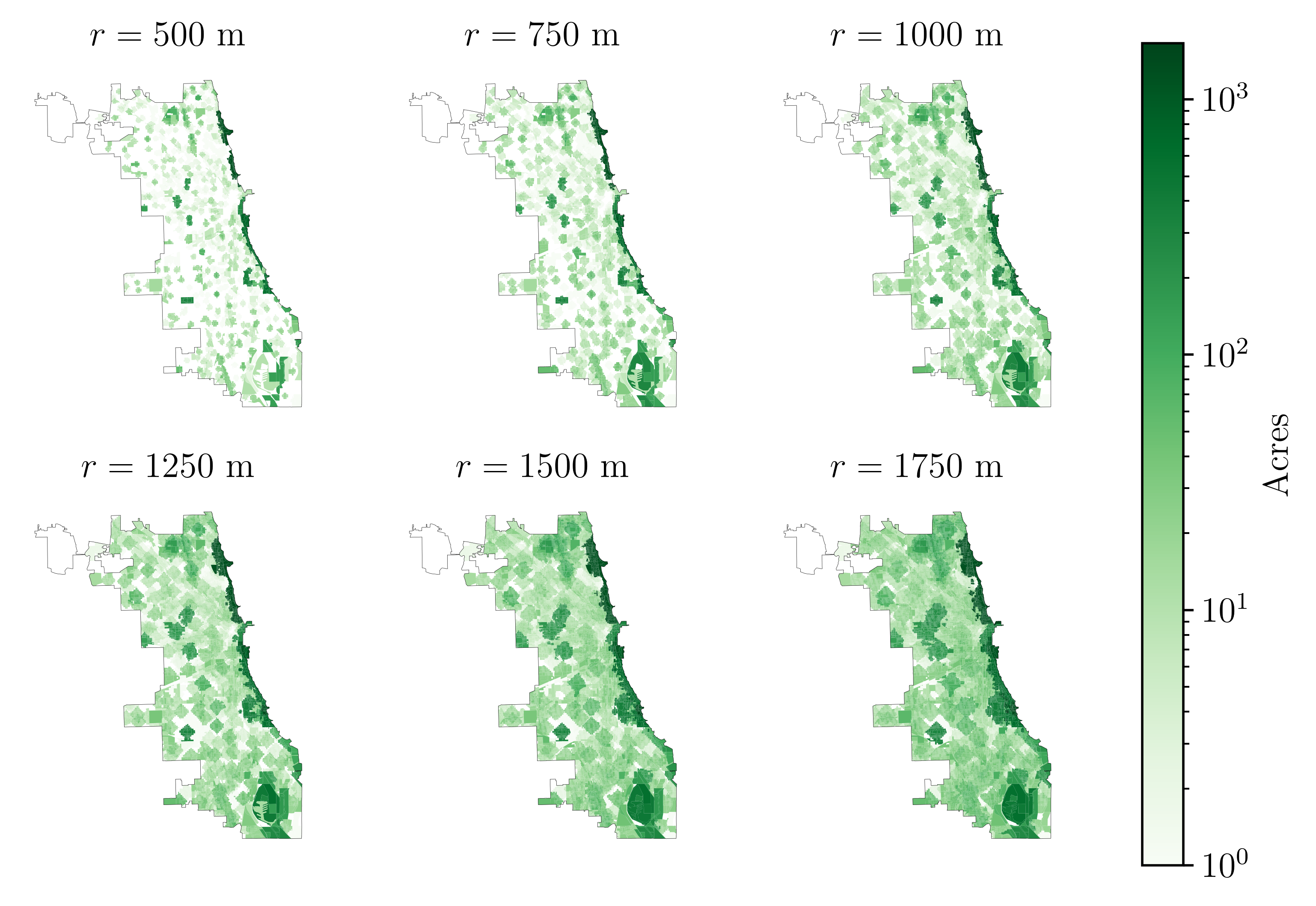}
\caption{We plot the {cumulative} quality score (i.e., total resource quality)
$f(\cdot, r): \mathcal{B} \to \R$ for $r \in \{500, 750, \ldots, 1750\}$ meters, where $\mathcal{B}$ is the set of census blocks of Chicago, the sites $\{M_i\}$ are public parks, and the score $s(M_i)$ is the acreage of the $i$th park.}
\label{fig:acre_block_scores_plot}
\end{figure}


\subsection{Bifiltration for Resource Sites}\label{sec:bifiltration_resource}

We now construct a bifiltration \linebreak $\{\X_{r, q}\}_{(r, q) \in \R^{\textnormal{op}} \times \R}$ that is indexed by the poset $\poset = \R^{\textnormal{op}} \times \R$.
 For each distance--quality pair $(r, q) \in \R^2$, we define the topological space
\begin{equation*}
    \X_{r, q} := \bigcup_{B_j  \in  \mathcal{B}_{r, q}} B_j\,,
\end{equation*}
where $\mathcal{B}_{r, q}$ is the set
\begin{equation}\label{eq:sub-level}
    \mathcal{B}_{r, q} := \{B_j \mid f(B_j, r) \leq q \}
\end{equation}
that we obtain by taking a sublevel set.
The space $\X_{r,q}$ consists of census units whose total quality within radius $r$ is at most
$q$. 
It identifies regions that fail to accumulate at least quality $q$ within the specified radius $r$. Census units that appear in $\mathcal{B}_{r, q}$ at larger $r$ values or smaller $q$ values tend to have less overall access to resource sites. 
In Figure~\ref{fig:bifiltration}, we show this bifiltration for our case study of parks.

\begin{figure}
    \includegraphics[width = .95\textwidth]{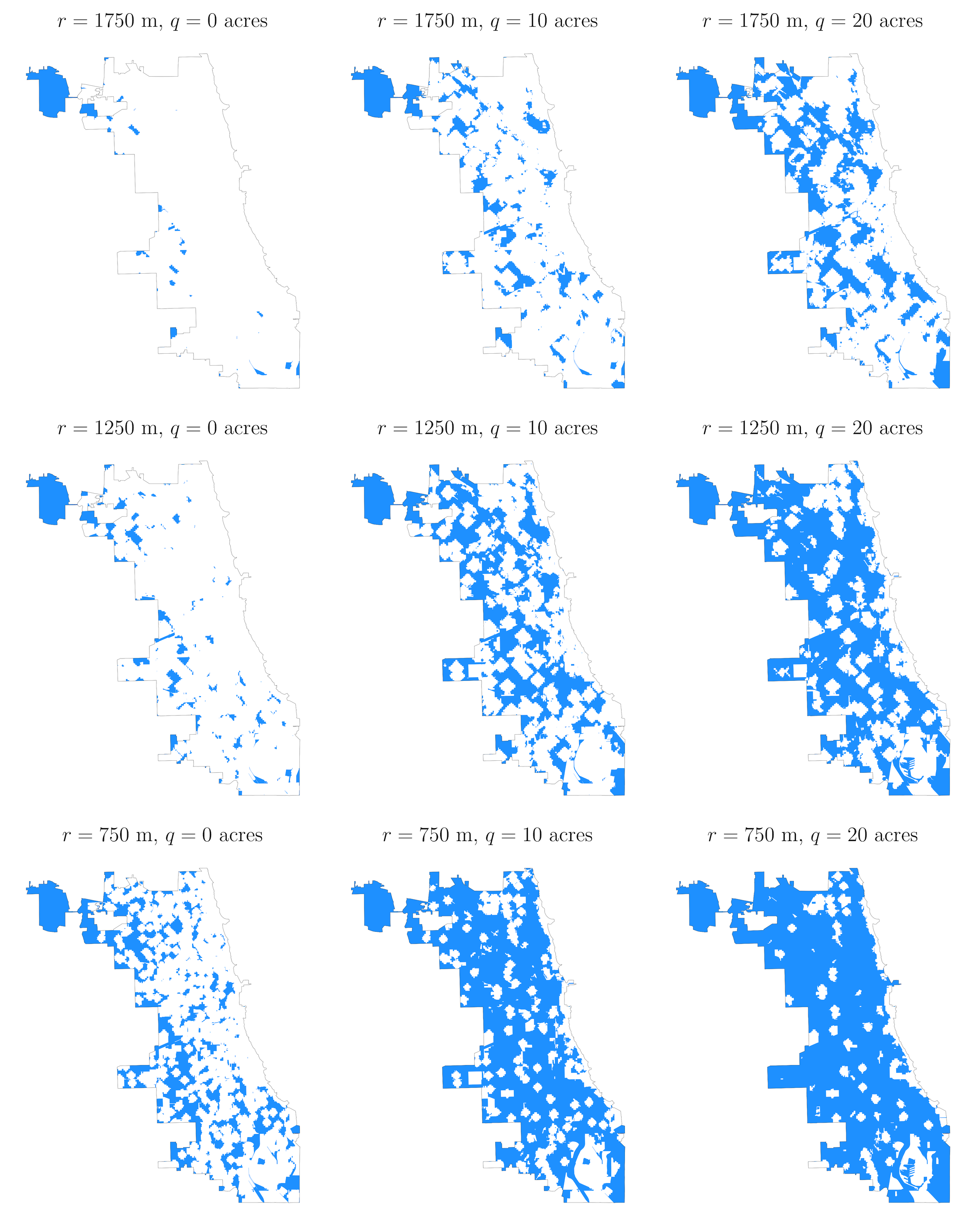}
    \caption{Our resource bifiltration for public parks in Chicago.
      In this example, $\mathcal{B}$ is the set of census blocks of Chicago, the sites $\{M_i\}$ are parks, and the quality score $s(M_i)$ is the acreage of the $i$th park. See Figure~\ref{fig:acre_block_scores_plot} for a visualization of the 
      associated {cumulative} resource quality
       $f:\mathcal{B} \times \R_{\geq 0} \to \R$.
       }
    \label{fig:bifiltration}
\end{figure}

\begin{lemma}\label{lem:bifiltration_verify} 
    The set $\{\X_{r, q}\}_{(r, q) \in \R^{\textnormal{op}} \times \R}$ is a bifiltration. That is, $\X_{r', q} \subseteq \X_{r, q'}$ if $r \leq r'$ and $q \leq q'$.
\end{lemma}

\begin{proof}
    For $B_j \in \mathcal{B}_{r', q}$, we have $B_j \in \mathcal{B}_{r, q'}$ because
    \begin{equation*}
        f(B_j, r) \leq f(B_j, r') \leq q \leq q'
    \end{equation*}
    {for all $r \leq r'$ and $q \leq q'$.}
\end{proof}

\subsection{Bifiltration for Nuisance Sites}\label{sec:bifiltration_nuisance}

Our bifiltration for nuisance sites is analogous to our bifiltration for resource sites, except that now $r$ is {increasing} and $q$ is decreasing. By contrast, for resource sites (see Section~\ref{sec:bifiltration_resource}), $r$ is {decreasing} and $q$ is {increasing}.
Accordingly, we now use superlevel sets of the {cumulative} score $f(\cdot, r)$ instead of sublevel sets.

We construct a bifiltration $\{\X_{r, q}\}_{(r, q) \in \R \times \R^{\textnormal{op}}}$ that is indexed by the poset $\poset = \R\times \R^{\textnormal{op}}$.
 For each distance--{severity} pair $(r, q) \in \R^2$, we define the topological space
\begin{equation*}
    \X_{r, q} := \bigcup_{B_j \in \mathcal{B}_{r, q}} B_j\,,
\end{equation*}
where $\mathcal{B}_{r, q}$ is the set
\begin{equation}\label{eq:super-level}
    \mathcal{B}_{r, q} := \{B_j \mid f(B_j, r) \geq q \}\,.
\end{equation}
The space $\mathcal{X}_{r, q}$ consists of census units whose total exposure within radius $r$ is at least $q$. 
Census units that appear in $\mathcal{B}_{r, q}$ at smaller $r$ values or larger $q$ values tend to have more overall exposure to nuisance sites.

\begin{lemma}\label{lem:bifiltration_nuisance} 
    The set $\{\X_{r, q}\}_{(r, q) \in \R \times \R^{\textnormal{op}}}$ is a bifiltration. That is, $\X_{r, q'} \subseteq \X_{r', q}$ if $r \leq r'$ and $q \leq q'$.
\end{lemma}

\begin{proof}
    For $B_j \in \mathcal{B}_{r, q'}$, we have $B_j \in \mathcal{B}_{r', q}$ because
    \begin{equation*}
        f(B_j, r') \geq f(B_j, r) \geq q' \geq q\,.
    \end{equation*}
    \vspace{0.01cm}
\end{proof}


\subsection{Interpretation of 0D Persistent Homology} \label{interpret}

The $0$D PH classes represent connected components. However, we interpret the components in the resource and nuisance bifiltrations in slightly different ways.
The components in the resource bifiltration $\mathcal{X}_{r, q}$ represent regions
where the total resource quality $f(B_j, r) \leq q$ for each census unit $B_j$ in the component. That is, census units $B_j$ in the component do not have access to at least a cumulative quality $q$ of resource sites within distance $r$. Therefore, the components that appear earliest in the bifiltration (i.e., at the largest $r$ values and smallest $q$ values) indicate the regions
with the poorest access to resources.
The components of the nuisance bifiltration $\mathcal{X}_{r, q}$ represent regions that have at least a cumulative exposure $q$ within radius $r$. The components that appear earliest in the bifiltration (i.e., at the smallest $r$ values and 
largest $q$ values) have the most exposure to nuisances.
For both resources and nuisances,
 the components (i.e., $0$D PH classes) represent the most disadvantaged regions.



\subsection{Computing our Bifiltration Using an Equivalent Graph Bifiltration}
\label{ssec:equiv_graphbifiltration}

Working with the bifiltration $\{X_{r,q}\}$ directly is computationally expensive, so we instead work with an equivalent bifiltration $\{G_{r, q}\}$ of graphs. 
Each graph $G_{r, q}$ is associated with the topological space $\mathcal{X}_{r, q}$. 
The vertices of the graph $G_{r, q}$ are
\begin{equation*}
    V_{r, q} := \{v_j \mid B_j \in \mathcal{X}_{r, q} \} \,,
\end{equation*}
and the edges of $G_{r, q}$ are
\begin{equation*}
    E_{r, q} := \{(v_j, v_k) \mid \text{census units } B_j, B_k \text{ are adjacent and } B_j, B_k \in \mathcal{X}_{r, q} \}\,.
\end{equation*}
The graph $G_{r, q}$ is a combinatorially simpler object than $\mathcal{X}_{r, q}$ because it is one-dimensional 
instead of two-dimensional.

The bifiltrations $\{X_{r, q}\}$ and $\{G_{r, q}\}$ have the same $0$D PH. 
There is a one-to-one correspondence between the components of $G_{r, q}$ and the components of $\mathcal{X}_{r, q}$
through the bijection
\begin{align*}
    \alpha^{r, q}: \mathcal{B}_{r, q} &\to V_{r, q}\,, \\
    B_j &\mapsto v_j\,.
\end{align*}
The function $\alpha^{r, q}$ induces an isomorphism $\alpha^{r, q}_*: H_0(\mathcal{X}_{r, q}) \to H_0(G_{r, q})$ that maps a component $C = \bigcup_k B_{j_k}$ in $\mathcal{X}_{r, q}$ to the corresponding graph component with vertices $\{v_{j_k}\}_k$. 
Each isomorphism $\alpha^{r, q}_*$ also induces
an isomorphism for 0D PH.

\begin{lemma}
    The $0$D PH
     of the bifiltration $\{\mathcal{X}_{r, q}\}$ is isomorphic to the $0$D PH
     of the bifiltration $\{G_{r, q}\}_{r, q}$.
\end{lemma}

\begin{proof}
    Consider the map $\alpha^{r, q} : \mathcal{B}_{r, q} \to V_{r, q}$ with
    $B_j \mapsto v_j$, and let $\alpha^{r, q}_*: H_0(\mathcal{X}_{r, q}) \to H_0(G_{r, q})$ be the induced isomorphism on $0$D homology for each distance--quality pair $(r, q)$. Let $\iota^{\mathcal{X}}:\mathcal{X}_{r, q} \hookrightarrow \mathcal{X}_{r', q'}$ be an inclusion, and let $\iota^G: G_{r, q} \hookrightarrow G_{r', q'}$ be the corresponding inclusion of graphs. The bijections $\alpha^{r, q}$ commute with the inclusions. That is, if $B_j$ is a census unit in $\mathcal{X}_{r, q}$, then
    \begin{equation*}
        \iota^G \circ \alpha^{r, q} (B_j) = v_j = \alpha^{r', q'} \circ \iota^{\mathcal{X}}(B_j)\,.
    \end{equation*}
    Therefore, the isomorphisms $\alpha^{r, q}_*$ on homology commute with the homomorphisms $\iota^{\mathcal{X}}_* : H_0(\mathcal{X}_{r, q}) \to H_0(\mathcal{X}_{r', q'})$ and $\iota^G_* : H_0(G_{r, q}) \to H_0(G_{r', q'})$ that are induced by the inclusions $\iota^{\mathcal{X}}$ and $\iota^G$, respectively. Consequently, $\{\alpha^{r, q}_*\}$ is an isomorphism of the 0D PH.
\end{proof}


\subsection{Biparameter Component Aggregation}\label{sec:aggregation}

A $0$D PH class in a distance--quality bifiltration is a collection of connected components of the following form.

\begin{definition}
    A \emph{$0$D PH class} of a distance--quality bifiltration is a set $\{C_{r,q}\}_{r, q}$ of connected components $C_{r,q} \subseteq \K_{r, q}$ such that $C_{r, q} = C_{r', q'} \cap \mathcal{X}_{r, q}$ whenever $\mathcal{X}_{r, q} \subseteq \mathcal{X}_{r', q'}$. 
\end{definition}
In Figure~\ref{fig:ph_class}, we show an example of a $0$D PH class.

\begin{figure}
\centering
    \includegraphics[width = \textwidth]{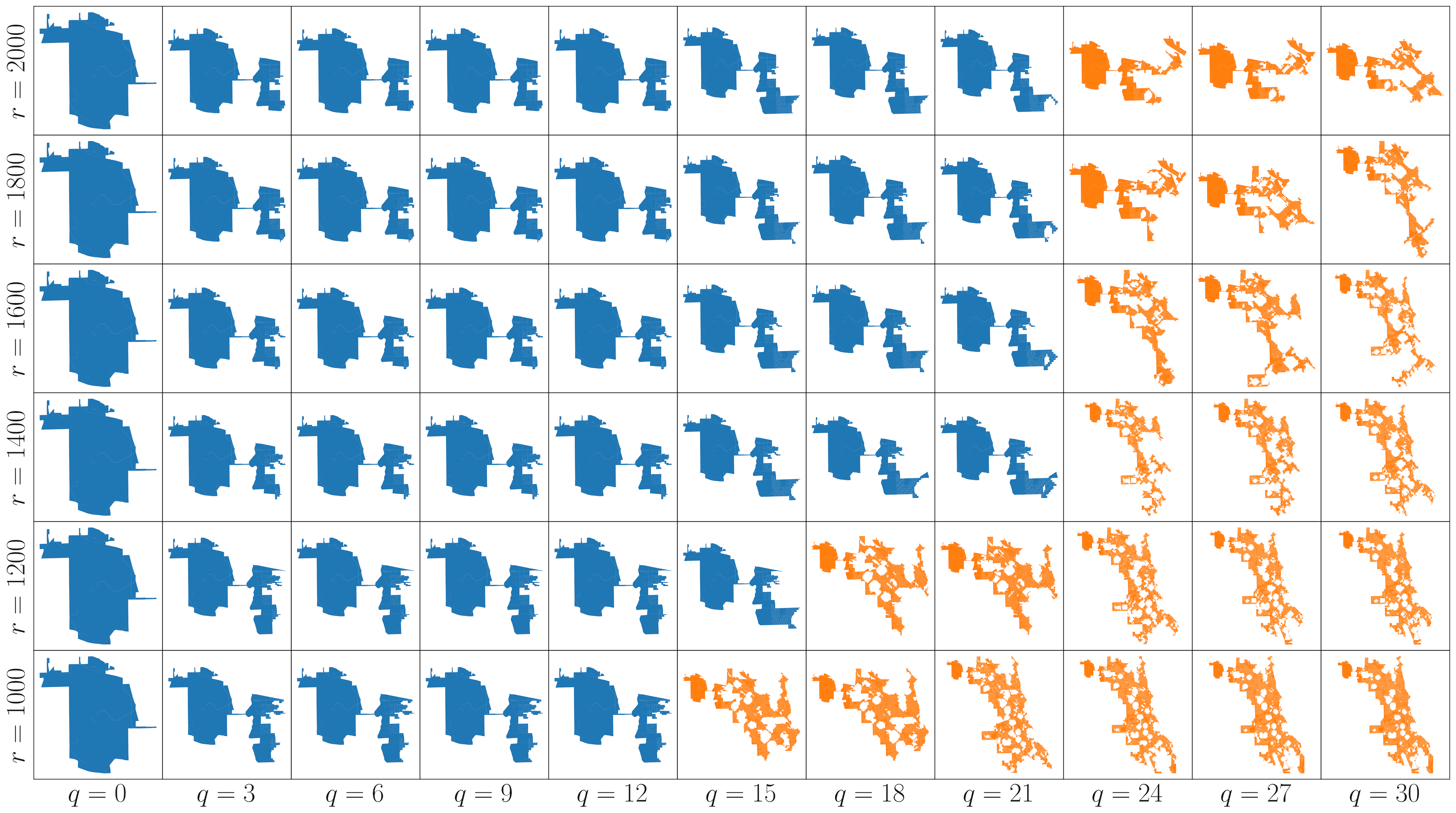}
    \caption{A $0$D PH class of a resource bifiltration $\{\mathcal{X}_{r, q}\}$. There is a connected component $C_{r, q}$ for each distance--quality pair $(r, q)$. In this example, $\mathcal{B}$ is the set of Chicago census blocks, the resource sites $\{M_i\}$ are parks, and the score $s(M_i)$ is the acreage of the $i$th park.
    The blue subplots indicate that $C_{r,q}$ is smaller than the chosen threshold (in this case,  1000 blocks). Orange subplots indicate $C_{r,q}$ is larger than the threshold. 
    }
    \label{fig:ph_class}
\end{figure}


\subsubsection{Truncated PH Classes}\label{truncate}

At late stages of a bifiltration, which entails small $r$ values and large $q$ values for a resource bifiltration and large $r$ values and small $q$ values for a nuisance bifiltration, a connected component $C_{r, q}$ can occupy a very large portion of the total 
area of a region.
 For example, in Figure~\ref{fig:ph_class}, the component $C_{r, q}$ includes more than $1000$ census blocks of Chicago for several values of the distance--quality pair $(r, q)$. By $(r, q) = (1200, 24)$, the component $C_{r, q}$ includes nearly all of Chicago. Such components are too large to be meaningful, so we define a ``truncated'' version of a PH class in which one only considers components with at most a specified number of census units.
 
\begin{definition}
    A \emph{truncated PH class} is a set $\{C_{r, q}\}_{r, q}$ such that there is a PH class $\{C'_{r, q}\}_{r, q}$ with
    \begin{equation*}
    C_{r, q} := 
        \begin{cases}
            C'_{r, q}\,, &\text{if } |C'_{r, q}| \leq S \\
            \emptyset \,, &\text{otherwise}
        \end{cases}
    \end{equation*}
    for some fixed
    size threshold $S$. We say that $\{C_{r,q}\}_{r,q}$ is a \emph{truncation} of $\{C'_{r,q}\}_{r,q}$.
\end{definition}
Alternatively, one can threshold the geographic area of the component $C_{r, q}$, rather than the number of census units.

In Figure~\ref{fig:truncated}, we show a truncated PH class in a bifiltration in Chicago with $S = 1000$.

\begin{figure}
    \centering
    \includegraphics[width = \textwidth]{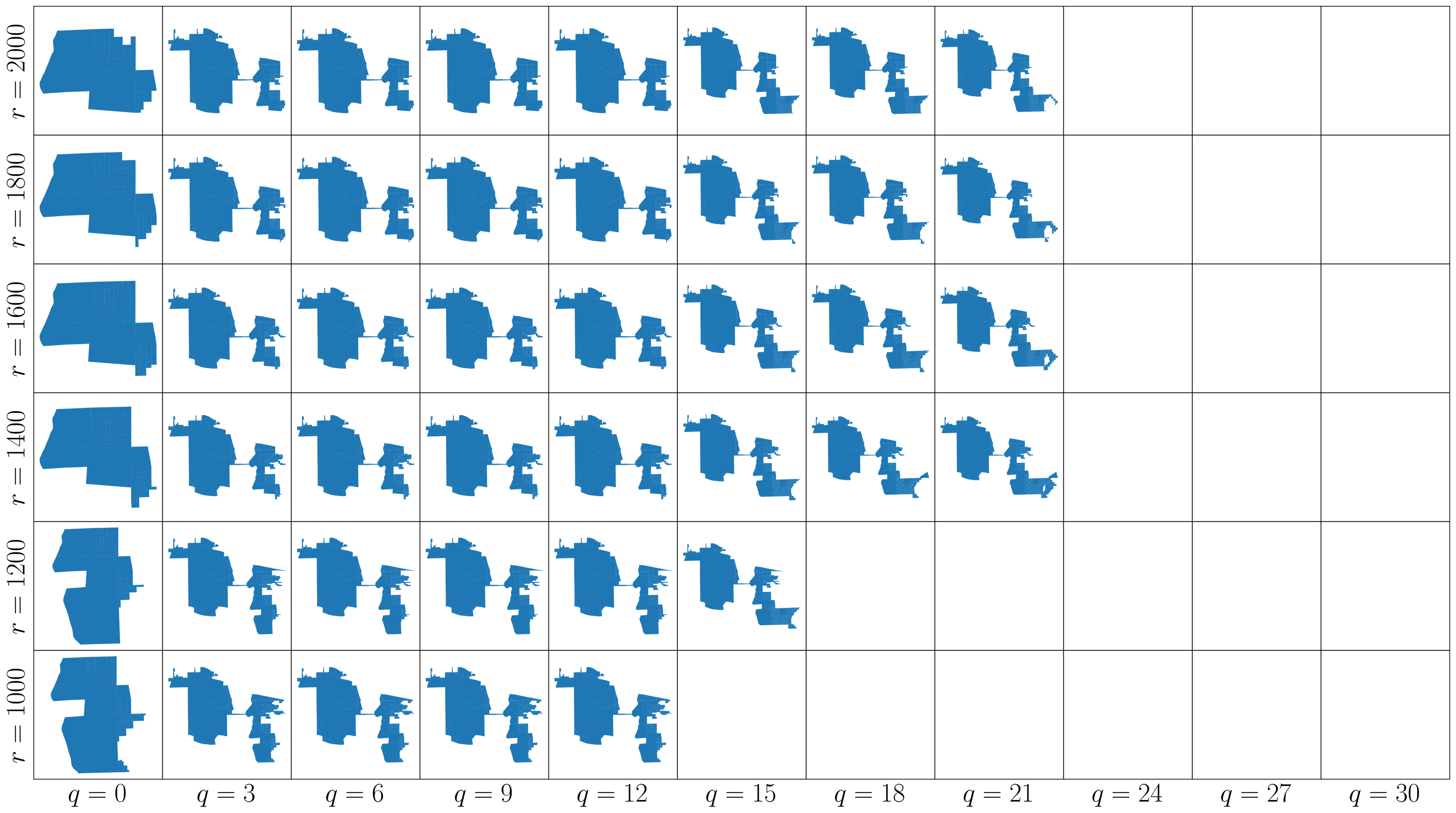}
    \caption{A $0$D truncated PH class of a resource bifiltration $\{X_{r, q}\}$ for public parks in Chicago. The quality score of each park is its acreage.
    The largest component has 
    $S = 1000$ census blocks. 
    This truncated PH class 
    {is not exactly the same as}
    the truncation of the PH class in Figure~\ref{fig:ph_class}, although {it is very similar to it}. {See the discussion in Section~\ref{similar}.}
    }
    \label{fig:truncated}
\end{figure}


\subsubsection{Algorithm to Aggregate Truncated PH Classes} \label{similar}
Multiple distinct truncated PH classes $\{C^1_{r,q}\}_{r,q}, \ldots, \{C^m_{r,q}\}_{r,q}$ can have very similar components. 
It is common
that $C^i_{r,q} = C^j_{r,q}$, especially at later stages of a bifiltration (i.e., small $r$ and large $q$ for resources and large $r$ and small $q$ for nuisances). For example, the truncated PH class in Figure~\ref{fig:truncated} is almost 
the same as the truncation of the PH class in Figure~\ref{fig:ph_class}; their components are the same except at radius $r = 0$. 
Intuitively, these two PH classes represent roughly ``the same'' hole in resource coverage.

Multiple truncated PH classes having the same components at later stages of a bifiltration is a common phenomenon. 
Early in a bifiltration
(i.e., when $r$ is large and $q$ is small for resources or when $r$ is small and $q$ is large for nuisances), the connected components of $\mathcal{X}_{r, q}$ tend to have very few census units. 
Many of these small components merge with
each other over the course of only a few steps in $r$ and/or $q$.
To handle such situations, we algorithmically
 {aggregate} truncated PH classes that are ``almost'' the same in some sense.

Our aggregation procedure (see Algorithm~\ref{alg:aggregation_procedure}) relies on pairwise comparisons of the ``terminal'' components of each truncated PH class. In our aggregation algorithm, we seek to merge truncated PH classes with the same or similar terminal components.

\begin{definition}
    Let $\{C_{r,q}\}_{r,q}$ be a truncated PH class. Its \emph{set of terminal parameter pairs} $\mathcal{T}$
    are the values of the distance--quality pair $(r, q)$ 
    for which $C_{r, q + 1} = \emptyset$. Its associated \emph{terminal components} are the components $C_{r,q}$ in which each $(r, q)\in \mathcal{T}$ is a terminal parameter pair.
\end{definition}
In Figure~\ref{fig:truncated}, the last component in each row is a terminal component and the values of the corresponding distance--quality pairs $(r,q)$ are the terminal parameters.

\begin{algorithm}
    \caption{Our algorithm to aggregate truncated PH classes.
    }
    \label{alg:aggregation_procedure}
    \begin{algorithmic}[1] 
    \Require Threshold $S$, bifiltration $\{\mathcal{X}_{r,q}\}_{r,q}$
    \Ensure A set $\{\mathcal{C}_{L_k}^\textnormal{agg}\}$ of aggregated PH classes    
        \State Calculate all distinct truncated PH classes $\{C_{r,q}^1\}_{r,q}, \ldots, \{C_{r,q}^m\}_{r,q}$
        
        \State Initialize a graph $H$ with a vertex $w_i$ for each $\{C_{r,q}^i\}_{r,q}$

        \For{each pair $(w_i, w_j)$ of vertices in $H$}
            \State Let $T_i$ and $T_j$ be the sets of terminal parameter pairs for $\{C_{r,q}^i\}_{r,q}$ and $\{C_{r,q}^j\}_{r,q}$, respectively

            \State Let $N_i$ and $N_j$ be the numbers of $(r, q)$ terminal parameter pairs in $T_i$ and $T_j$, respectively, with
            $C_{r,q}^i = C_{r,q}^j$

            \If{$N_i \geq .5 |T_i|$ or $N_j \geq .5|T_j|$}
                \State Add the edge $(w_i, w_j)$ to $H$
            \EndIf
        \EndFor

        \State Calculate the connected components $K_1, \ldots, K_n$ of $H$

        \For{each component $K_i$}
            \State Calculate the provisional aggregated PH class 
            $\widetilde{\mathcal{C}}_i^\textnormal{agg} := \{\widetilde{C}^{\textnormal{agg}}_{r,q}\}_{r,q}$ for the truncated PH classes $\{C^{i_1}_{r, q}\}_{r, q}, \ldots, \{C^{i_\ell}_{r, q}\}_{r, q}$, where $\{w_{i_j}\}_{j=1}^\ell$ are the vertices of $K_i$
        \EndFor

        \For{each $\widetilde{\mathcal{C}}_i^\textnormal{agg}$}
            \State Calculate the maximal component $\widetilde{\mathcal{M}}_i$
        \EndFor

        \State Initialize a graph G with a vertex $v_i$ for each $\{\widetilde{\mathcal{M}}_i\}_i$

        \For{{each pair $(v_i, v_j)$ in $G$}}
        \If{{if $|\widetilde{\mathcal{M}}_i| \cap |\widetilde{\mathcal{M}}_j| \geq 0.9|\widetilde{\mathcal{M}}_i|$ or $|\widetilde{\mathcal{M}}_i| \cap |\widetilde{\mathcal{M}}_j| \geq 0.9|\widetilde{\mathcal{M}}_j|$}}
        \State{Add the edge $(v_i,v_j)$ to $G$}
        \EndIf
        \EndFor

        \State{Calculate the connected components $L_1, \ldots, L_m$} of $G$

        \For{{each component $L_k$}}
        \State{Merge $\widetilde{\mathcal{C}}_{k_1}^\textnormal{agg}, \ldots, \widetilde{\mathcal{C}}_{k_\ell}^\textnormal{agg}$ where $\{v_{k_i}\}$ are the vertices of $K_i$ into a final aggregated PH class $\widetilde{\mathcal{C}}_{L_k}^{\textnormal{agg}}$}
        \EndFor
        
        \State \Return The set $\{\mathcal{C}_{L_k}^\textnormal{agg}\}$ of aggregated PH classes
    \end{algorithmic}

\end{algorithm}

We now give the details of our aggregation procedure, which we present as pseudocode in Algorithm~\ref{alg:aggregation_procedure}. We first set a threshold $S$ and calculate all distinct truncated PH classes $\{C_{r,q}^1\}_{r,q}, \ldots, \{C_{r,q}^m\}_{r,q}$ for the threshold $S$. We then construct a graph $H$ that has
a vertex $w_i$ for each truncated PH class $\{C_{r,q}^i\}_{r,q}$. We add an edge $(w_i, w_j)$ if either $C_{r,q}^i = C_{r,q}^j$ for more than half of the terminal parameter pairs $(r, q)$ of 
$\{C_{r,q}^i\}_{r,q}$ or $C_{r,q}^i = C_{r,q}^j$ for at least half of the terminal parameter pairs of $\{C_{r,q}^j\}_{r,q}$. 
The two conditions are not equivalent because the $i$th and $j$th truncated PH classes 
usually
have different sets and different numbers of terminal parameter pairs.
If a set $\{w_i\}$  of vertices forms a connected component of $H$, we merge their corresponding truncated PH classes to make an \emph{aggregated truncated PH class}.
For simplicity, we will simply refer to this as an \emph{aggregated PH class}.
Performing this merging procedure
for each component of $H$ yields a provisional set of aggregated PH classes.

\begin{definition}
    Let $\{C^{i_1}_{r, q}\}_{r, q}, \ldots, \{C^{i_\ell}_{r, q}\}_{r, q}$ be a set of $\ell$ truncated PH classes. Their \emph{aggregated PH class} is the set $\{C^\textnormal{agg}_{r, q}\}_{r, q}$, with
    \begin{equation*}
        C^\textnormal{agg}_{r, q} := \bigcup_{j = 1}^\ell C^{i_j}_{r, q}\,.
    \end{equation*}
\end{definition}
Figure~\ref{fig:aggregated_truncated_ph_classes} shows an example of an aggregated PH class.

\begin{figure}
    \centering
    \includegraphics[width=\linewidth]{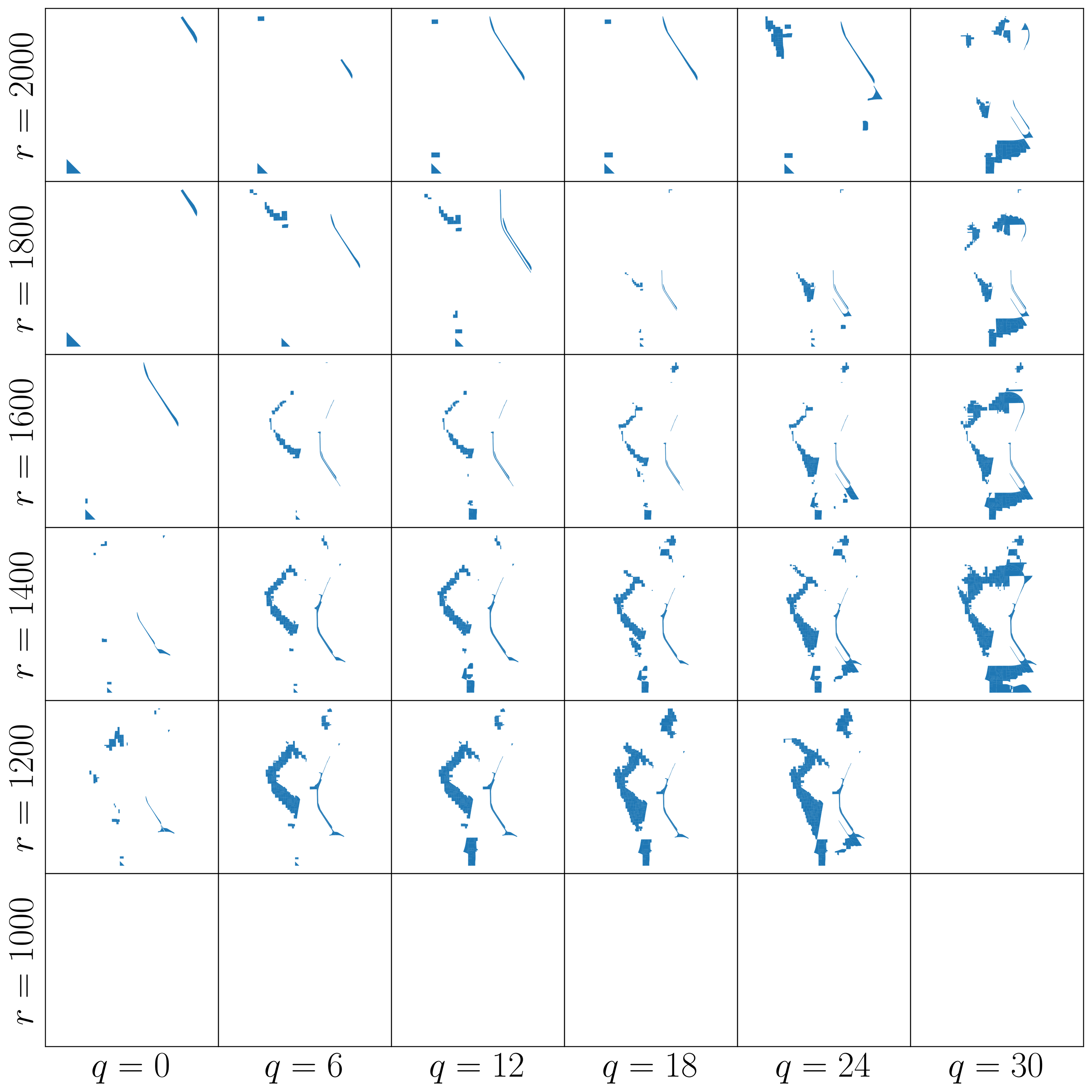}
    \caption{ An aggregated PH class $\mathcal{C}^\textnormal{agg}$ for the resource bifiltration that we construct for {our case study of public parks in Chicago.} At many values of the {distance--quality pair $(r, q)$,} the set $\mathcal{C}^\textnormal{agg}_{r, q}$ has multiple components.    }
    \label{fig:aggregated_truncated_ph_classes}
\end{figure}

Algorithm~\ref{alg:aggregation_procedure} includes a final step in which we merge some of the provisional aggregated PH classes. We do these merges by extracting a single \emph{maximal component} for each provisional aggregated PH class and then comparing these maximal components.

\begin{definition}
    Let $\{C^\textnormal{agg}_{r,q}\}_{r,q}$ be an aggregated PH class with terminal parameter pairs $\mathcal{T}$. The \emph{maximal component} is $\bigcup_{(r, q) \in \mathcal{T}} C^\textnormal{agg}_{r,q}$\,.
\end{definition}

The maximal component of a provisional aggregated PH class $\widetilde{\mathcal{C}}_1^{\textnormal{agg}}$ 
sometimes has significant overlap with the maximal component of another aggregated PH class $\widetilde{\mathcal{C}}_2^{\textnormal{agg}}$. This situation can arise when the terminal parameter pairs of $\widetilde{\mathcal{C}}_1^{\textnormal{agg}}$ are disjoint from those of 
$\widetilde{\mathcal{C}}_2^{\textnormal{agg}}$. See Figure~\ref{fig:final_merge_step} for an example where the maximal component of one aggregated PH class is contained completely within the maximal component of another aggregated PH class. 
In situations where the overlap surpasses some threshold, we merge the aggregated PH classes.  
More precisely, if 90\% of one maximal component is contained within another maximal component, we merge the aggregated PH classes $\widetilde{\mathcal{C}}_1^{\textnormal{agg}}$ and $\widetilde{\mathcal{C}}_2^{\textnormal{agg}}$ into a larger, final aggregated PH class $\mathcal{C}^{\textnormal{agg}}$, with
\begin{equation*}
    \mathcal{C}^{\textnormal{agg}}_{r,q} = (\widetilde{\mathcal{C}}_1^{\textnormal{agg}})_{r,q} \cup (\widetilde{\mathcal{C}}_2^{\textnormal{agg}})_{r,q}\,.
\end{equation*}
In order to ensure the order of merging does not effect the outcome, we do this by identifying sets of maximal components of aggregated PH classes where they satisfy the 90\% overlap condition with at least one other maximal component in the set. 
We then merge all aggregated PH classes in the set. 
This does lead to regions larger than the truncation threshold, $S$. 
In theory, it could lead to very large regions, but in practice we find the merged components to be fairly localized. 
Completing this final merging step yields the final aggregated PH classes.

\begin{figure}
    \centering
    \includegraphics[width=0.95\linewidth]{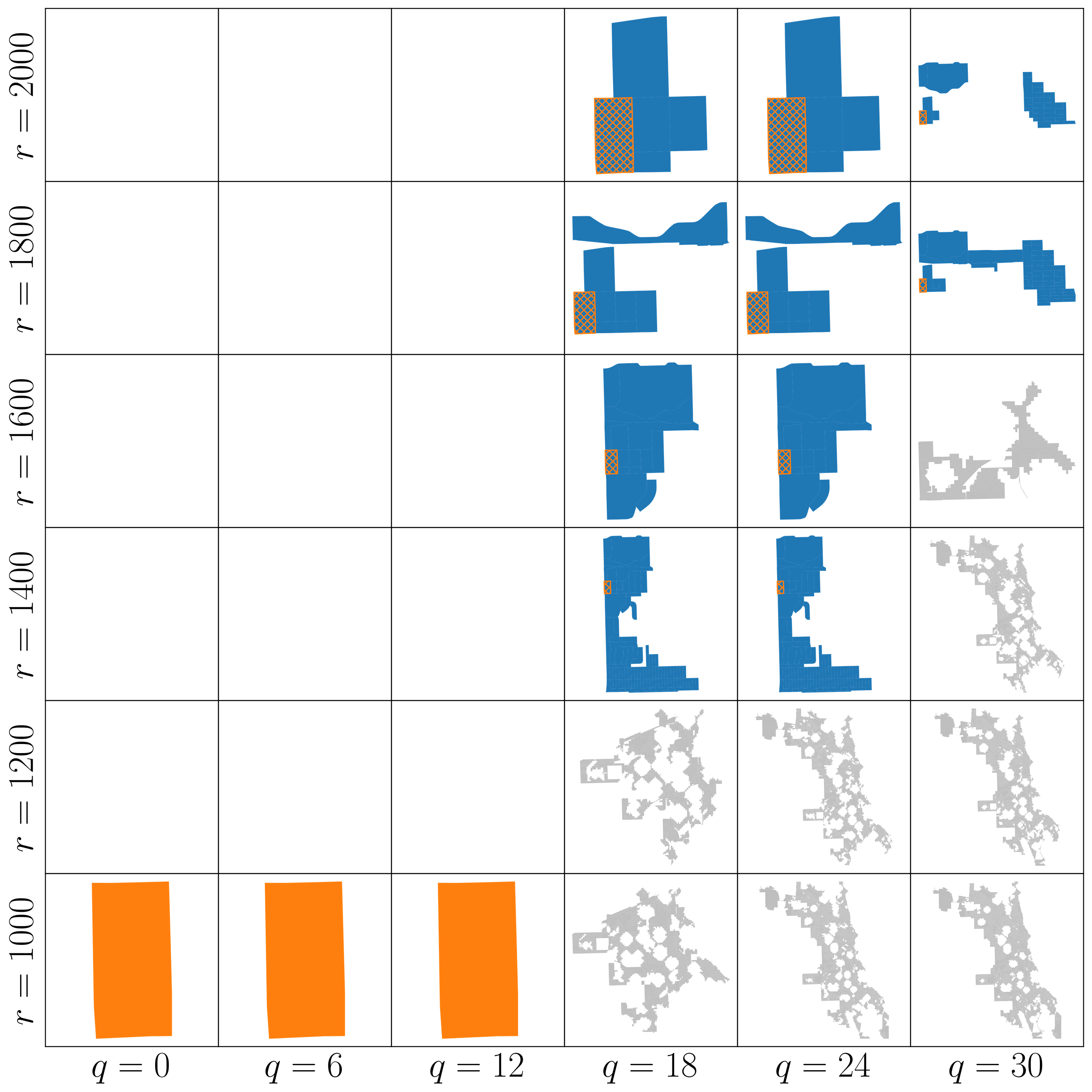}
    \caption{
    We show an example of PH classes from our case study to parks that are merged in the last step of Algorithm~\ref{alg:aggregation_procedure} (Lines 17--21) due to the fact that their terminal parameter pairs are disjoint.
    In this figure, there are two PH classes shown, one drawn in orange and one drawn in blue. 
    The grey subplots indicate those where $C_{r,q}$ is larger than the preset threshold (set to 1000 blocks) and the orange crosshatch region shown on the blue PH class is the block from the orange PH class. 
    Because the maximal component for the orange PH class is a subset of the maximal component from the blue PH class, so we merge the PH classes in the final step of Algorithm \ref{alg:aggregation_procedure}. 
     }
    \label{fig:final_merge_step}
\end{figure}


\subsubsection{The Lifetime of an Aggregated PH Class} \label{sec:range}

To quantify the {persistence of} aggregated PH classes, we define the \emph{lifetime} of an aggregated PH class. {Our notion of a lifetime is inspired by the notion of persistence of one-parameter PH classes \cite{otter2017}.} 
In {one-parameter PH}, persistence quantifies the length of {an interval that a PH class exists} before it merges into another class. 
{Analogously,} the lifetime of an aggregated PH class in {a bifiltration} measures the area {that an aggregated PH class occupies in the distance--quality plane of aggregated PH class, subject to the condition (see Section \ref{truncate}) that it is not too large.}

The lifetime {(see the definition below)} of an aggregated PH class 
{measures} the proportion of {bifiltration-parameter values that the class is nonempty and not too large.}
For each bifiltration, we {construct} a finite grid $\{r_1, \ldots, r_m\}\times \{q_1, \ldots, q_m\}$ of {bifiltration-parameter} values, which we use to calculate the lifetime of each aggregated PH class in that bifiltration.

\begin{definition}
    Let $\{r_1, \ldots, r_m\}\times \{q_1, \ldots, q_m\}$ be parameter values 
    in a bifiltration.
    The \emph{lifetime} (with respect to those parameters) of an aggregated PH class $\mathcal{C}^\textnormal{agg}$ is
    \begin{equation*}
        \frac{\vert\{(r_i, q_j) \mid \mathcal{C}^\textnormal{agg}_{r_i, q_j} \neq \emptyset\}\vert}{m^2}\,.
    \end{equation*}
\end{definition}


\section{Our Case Studies}\label{sec:results}

We compute aggregated PH classes and use a bifiltration to determine maximal components for the public parks, landfills, and bars and pubs. 
See Table~\ref{tab:parameters} for details about the parameter values that we use in our case studies.

\begin{table}[]
    \centering
    \begin{tabular}{ | m{7.4em} | m{9em}| m{9em} | m{6.5em} | } 
        \hline 
        \textbf{Sites} & \textbf{Distance ($r$)} & \textbf{Score ($q$)} & \textbf{Size Threshold ($S$)} \\
        \hline 
        Parks (resource) & 2000, 1800, 1600, 1400, 1200, 1000 meters & 0, 6, 12, 18, 24, 30 acres & 1000 blocks\\
        \hline 
        Landfills (nuisance) & 10, 15, 20, 25, 30, 35, 40 kilometers & 5, 15, 25, 35, 45, 55 short tons of waste & 10\% of the area\\
        \hline 
        Pubs (resource) & 2000, 1700, 1400, 1100,  800,  500 meters & 400,  700, 1000, 1300, 1600, 1900 stars & 20\% of the area \\
        \hline 
        Pubs (nuisance) & 500,  800, 1100, 1400, 1700, 2000 meters & 1000, 1200, 1400, 1600, 1800, 2000 reviews & 20\% of the area \\
        \hline 
    \end{tabular}
    \caption{
    The parameter values that we use for our case studies.}
    \label{tab:parameters}
\end{table}

The maximal components represent holes in coverage of resources and high-exposure regions for nuisances. 
They thus indicate disadvantaged regions.
A maximal component with a longer lifetime indicates that a region has
 less coverage or more exposure 
 within a geographically 
local
 area.
We interpret the maximal components with the longest lifetimes as the most problematic regions, either in terms of insufficiency of resources or overexposure to nuisances. Therefore, we focus on the maximal components with the longest lifetimes.

For each case study, the bifiltration $\X_{r,q}$ depends on the function $f(B_j,r)$, which in turn depends on the score $s(\M_i)$ and the distance $d(B_j,M_i)$. 
We customize each score to account for relevant features of each resource or nuisance.
In Table~\ref{tab:summary}, we summarize our distance measures $d(B_j,M_i)$ and quality/severity scores $s(\M_i)$ for our case studies.
We discuss our case study of public parks in Section~\ref{sec:parks}, our case study of landfills in Section~\ref{sec:landfills}, and our case study of bars and pubs in Section~\ref{sec:pubs}.

\begin{table}  
\centering 
\begin{tabular}{|l |l| l|}
    \hline 
    \textbf{Sites} & \textbf{Distance} & \textbf{Score}  \\ 
    \hline
     Parks (resource) & $L_1$ & acres \\ 
     \hline
     Landfills (nuisance) & $L_2$ & amount of waste (in short tons) \\ 
     \hline
     Pubs (resource) & $L_1$ & total stars \\ 
     \hline
     Pubs (nuisance) & $L_2$ & number of ratings  \\ 
     \hline
\end{tabular}
\caption{Summary of the examined resources and nuisances, the distances $d(B_j,M_i)$,
and the scores $s(M_i)$. }
\label{tab:summary}
\end{table}


\subsection{Public Parks in Chicago}\label{sec:parks}  

In our first case study, we examine public parks in Chicago. The geographic region $D$ is the city of Chicago, the set $\mathcal{B}$ of geographic units is the set of census blocks, and the sites $\{M_i\}$ are public parks.


\subsubsection{Score} 

One can determine the quality score of a park in many different ways\footnote{For example, see  \cite{wu2026ecosystemservicedemand} for additional important park features that one can use to determine quality.},
and how one measures park quality has a major
effect on our qualitative results.
For simplicity,
in our case study, we measure quality as the total acreage of a park.\footnote{See our Github repository to explore other notions of park quality \cite{ourgithub}. 
}
In Figure~\ref{fig:park_scores}, we show the parks and color them by their sizes.

\begin{figure}
    \centering
    \includegraphics[width=0.48\linewidth]{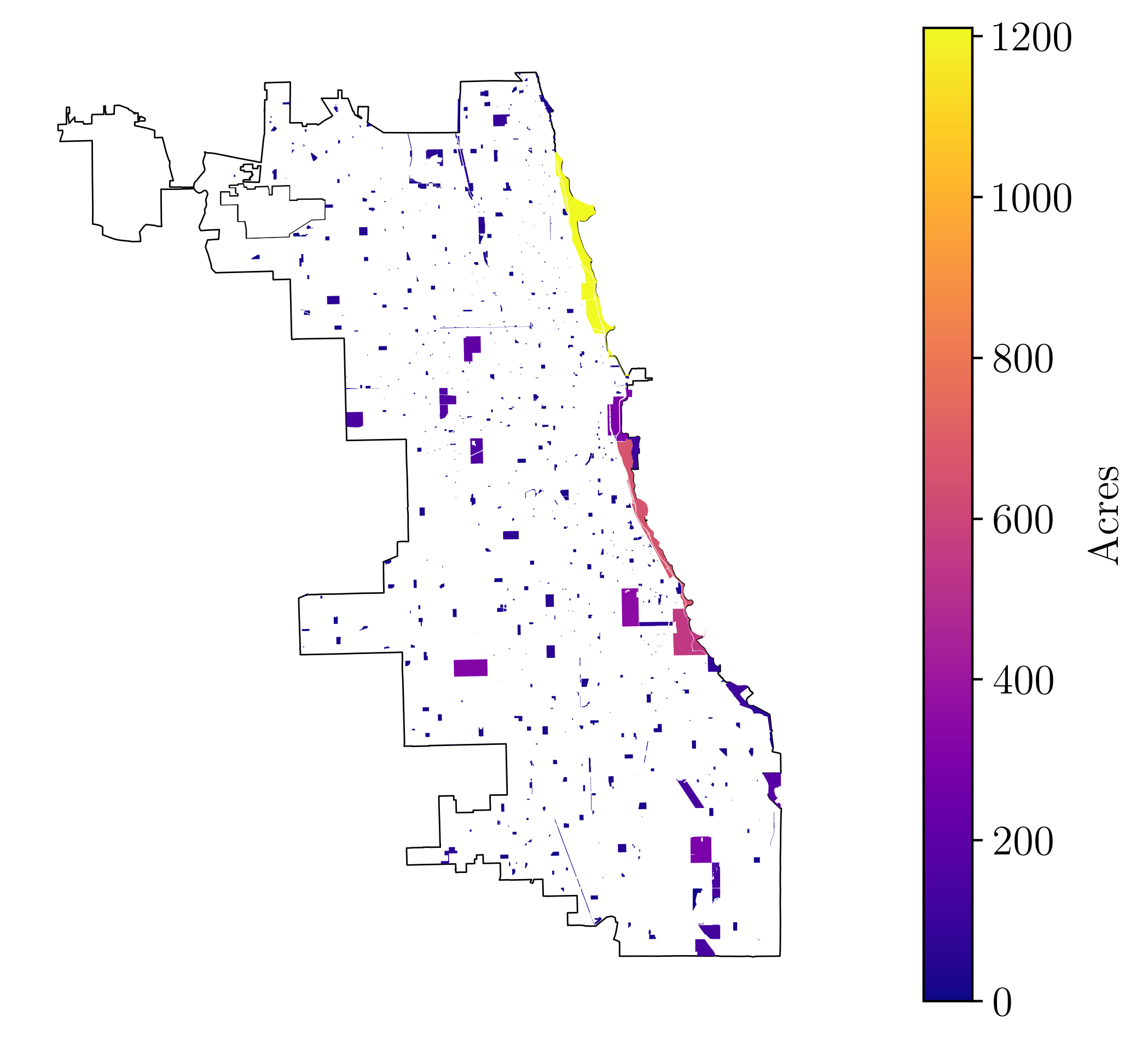}
    \caption{Public parks in Chicago. We color them by their sizes.}
    \label{fig:park_scores}
\end{figure}


\subsubsection{Distance} 

To quantify access to parks, we examine distances from home locations to parks via city streets. It seems desirable to use a tool like OpenStreetMap \cite{openstreetmap} or Google Maps \cite{GoogleMapsAPI} to compute distances between census blocks and parks. However, due to both computational cost and monetary cost, we instead simply use the 
$L_1$ distance.
When the city streets are grid-like, as they are in Chicago, the $L_1$ distance is a good approximation of the city-street distance.


\subsubsection{Results}

In Figure~\ref{fig:Parks_Results}, we show the maximal components for Chicago parks with the six
longest lifetimes.
These components span various parts of the city, although the five components with the longest lifetimes are all in the southern half of Chicago. 
These 
regions are the ones that require people to travel
the farthest
to access reasonable-quality parks.

\begin{figure}
    \centering
    \includegraphics[width=0.45\linewidth]{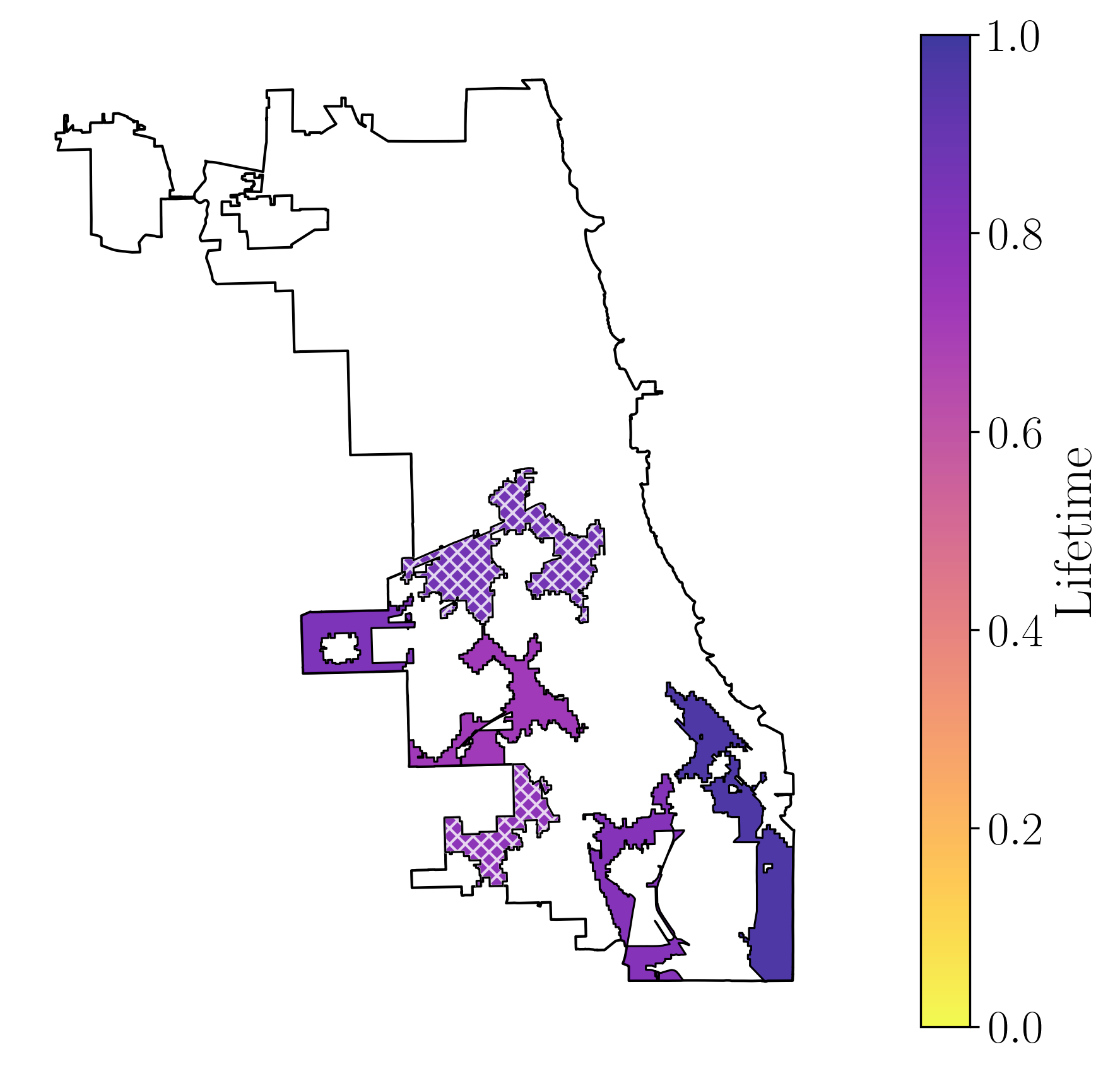}
    \caption{
    The maximal components for the Chicago parks with the 6 longest lifetime (see
     Section~\ref{sec:range}). 
    The cross hatching is used to help distinguish between regions.
    }
    \label{fig:Parks_Results}
\end{figure}


\subsubsection{Comparison to Our One-Parameter PH Approach}

In Section~\ref{sec:single_baseline}, we used a naive one-parameter filtration to identify the geographic regions that are farthest from public parks. That sublevel-set filtration does not account for park quality. 
In Figure~\ref{fig:Parks_Superlevel_Filtration}, we show four steps in the sublevel-set filtration. The parks are distributed roughly uniformly across Chicago, so most of the components have very short lifetimes. The components quickly merge into one large component that includes
most of Chicago. 
Consequently, this one-parameter approach is not very informative.

From the one-parameter filtration, we identify components with sizes that are less than
a specified threshold, analogously
to how we constructed
truncated PH classes in our bifiltration.
In Figure~\ref{fig:parks_baseline_components}, we 
show
the resulting components.
Many of the components from the 
sublevel-set 
filtration
simply captures regions that are far from resource sites.
Our bifiltration provides a much more detailed
and nuanced picture of park coverage.

\begin{figure}
    \centering
       \includegraphics[width = .4\textwidth]{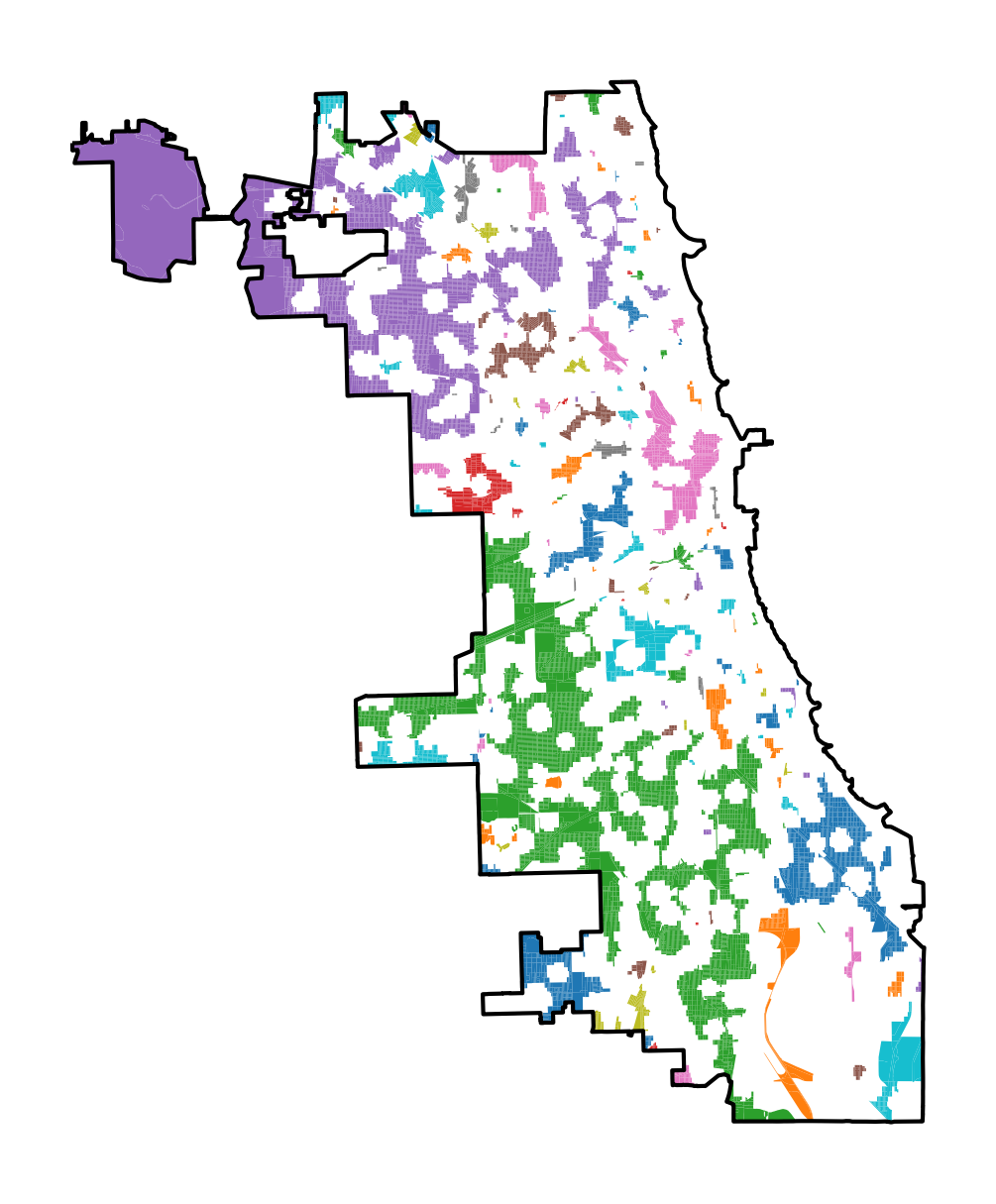}
    \caption{Components in our naive one-parameter superlevel-set filtration for accessibility of public parks in Chicago. 
    We randomly assign colors to help distinguish visually between neighboring components.
    See Figure~\ref{fig:Parks_Superlevel_Filtration} for a visualization of this superlevel-set filtration.
    }    
    \label{fig:parks_baseline_components}
\end{figure}


\subsection{Landfills in Illinois}\label{sec:landfills}

In our second case study, we examine landfills in Illinois. The geographic region $D$ is the state of Illinois, the sites $\{M_i\}$ are landfills, and
the set $\mathcal{B}$ is the set of census tracts. Census tracts, which are aggregations of census blocks,
have 
meaningful variance in their landfill exposure at the state scale. We use the coarser geographic unit of census tracts rather than census blocks because the latter are computationally demanding for a statewide analysis.


\subsubsection{Score} 

The score of each landfill is the total amount of waste (in short tons\footnote{A short ton is $2000$ lbs, which is approximately $907$ kg.}) at that landfill. 
Environmental hazards do not always scale directly with the amount of waste, but intuitively we expect that the
harmful chemicals
that arise from waste decomposition
are 
correlated with waste mass.
In Figure~\ref{fig:landfills}, we show the landfills and color them by their amount of waste (in short tons).

\begin{figure}
    \centering
    \includegraphics[width=0.4\linewidth]{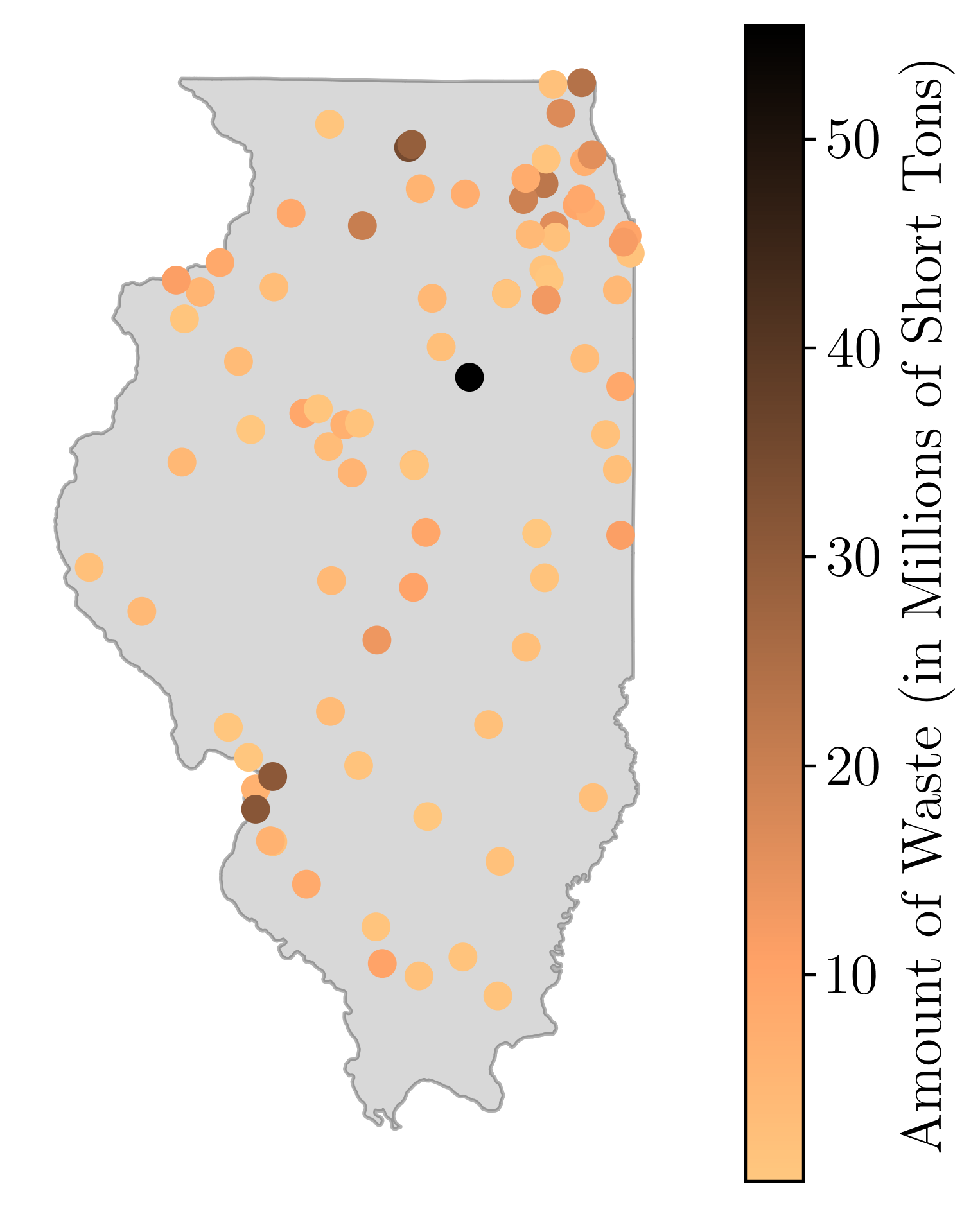}
    \caption{ Landfills in the state of Illinois. We color them according to their amount of waste (in millions of short tons). 
    }
    \label{fig:landfills}
\end{figure}


\subsubsection{Distance} 

To estimate the exposure to toxins and pollutants from a landfill (e.g., through air or groundwater), we let $d(M_i, T_j)$ be the Euclidean distance between landfill $M_i$ and the centroid of census tract $B_j$. 


\subsubsection{Results}

In Figure~\ref{fig:Landfill_Results}a, we show the maximal components for our case study of landfills.
Figure~\ref{fig:Landfill_Results}b shows the four maximal components with the largest lifetimes.
The most-exposed region is in
the southwest.
By comparing Figure~\ref{fig:Landfill_Results}a to Figure~\ref{fig:landfills}, we observe that this region neither has
the highest landfill density nor includes the individual
 landfill with the most waste. 
The two regions --- 
the Chicago area (i.e., around the northeastern part of Illinois) and slightly southwest of Chicago
--- that one may intuitively expect to have the most
exposure based on independently treating the locations and waste amount 
are not the most exposed regions when we analyze these two factors simultaneously. 
This observation highlights an important benefit of multiparameter PH.

\begin{figure}
    \centering
    \subfloat[Components from bifiltration]{\includegraphics[height = 50mm]{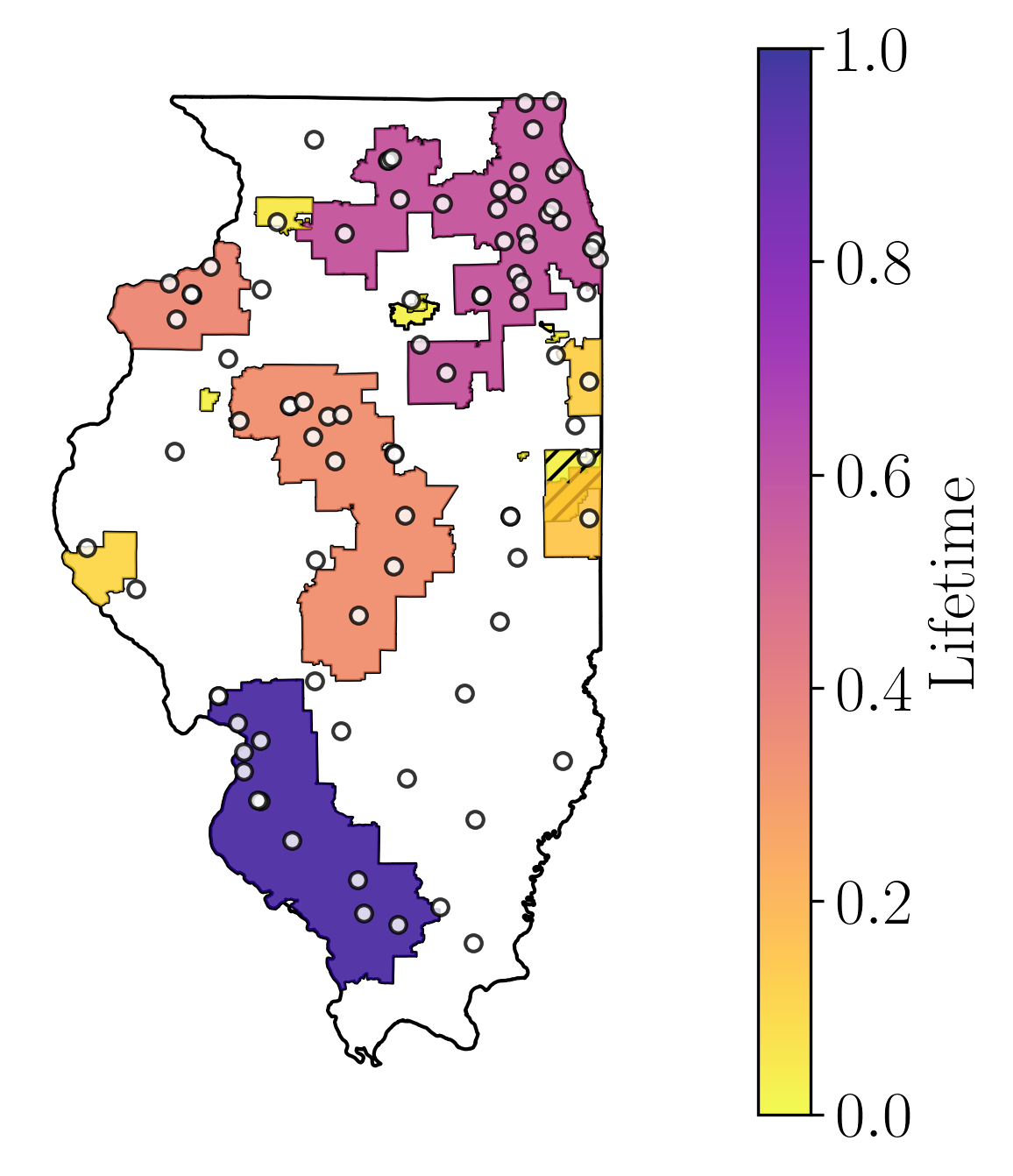}}
    \hspace{5mm}
    \subfloat[Top 4 components from bifiltration]{\includegraphics[height = 50mm]{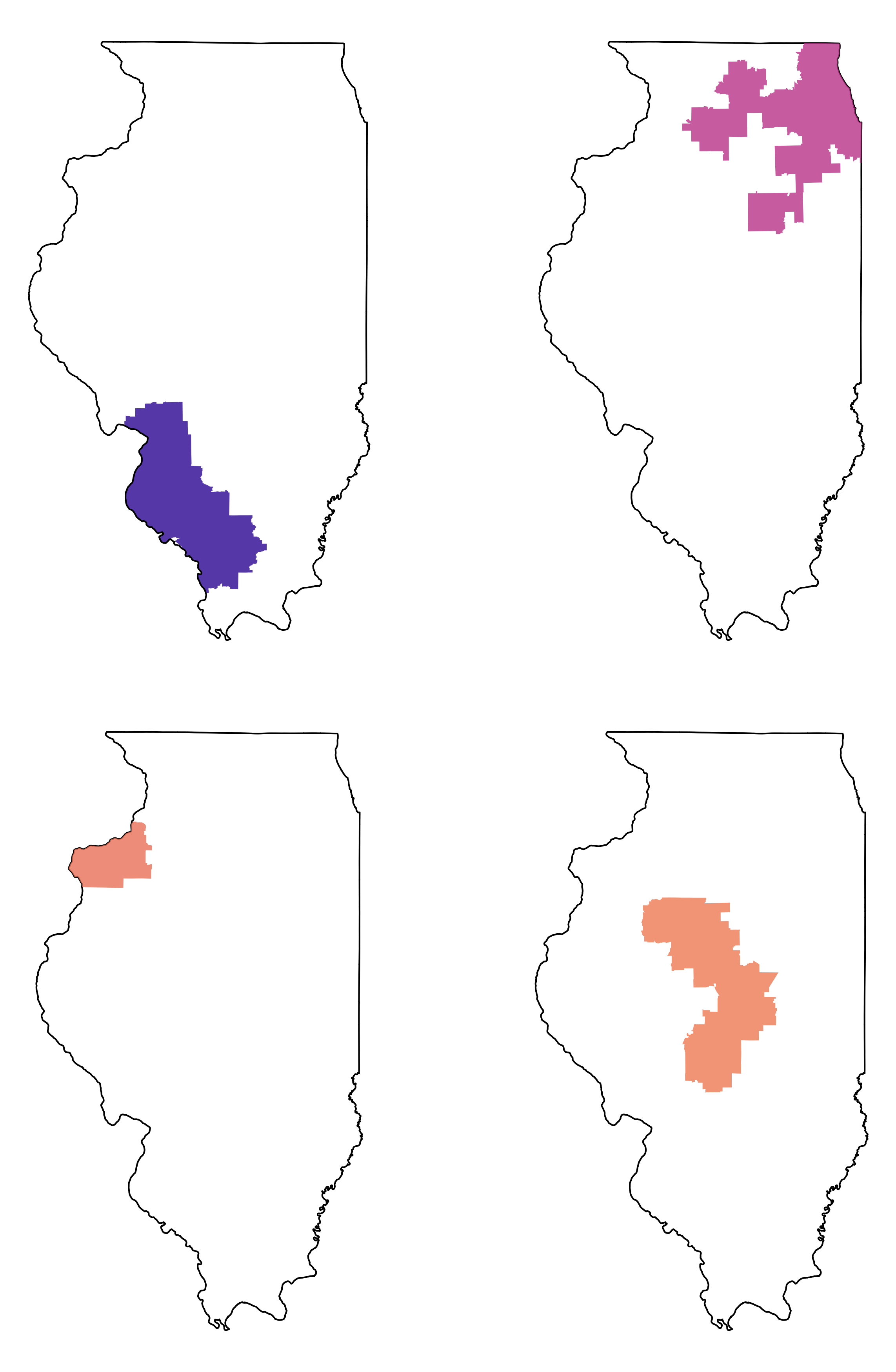}}
    \hspace{5mm}
    \subfloat[Components from one-parameter filtration]{\includegraphics[height = 50mm]{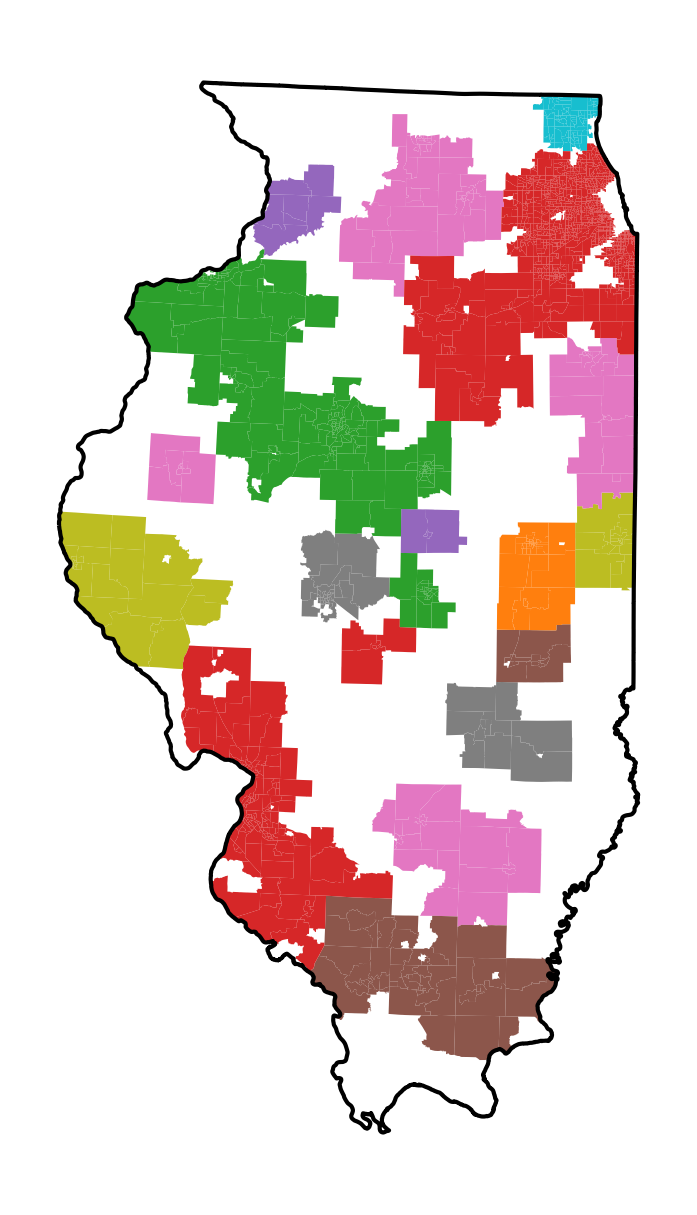}}
    \caption{(a) Maximal components for landfills (black dots). We color the components by their lifetimes (see
     Section~\ref{sec:range}). The cross hatching is used to help distinguish between overlapping regions. 
     (b) The four components from (a) with the largest lifetimes plotted individually for clarity. (c) 
     The components of our naive one-parameter superlevel-set filtration based on distance to the nearest landfill. 
     We randomly assign colors to help distinguish visually between neighboring components.}
    \label{fig:Landfill_Results}
\end{figure}


\subsubsection{Comparison to Our One-Parameter PH Approach}

In Figure~\ref{fig:Landfills_Sublevel_Filtration}, we show four steps in the naive one-parameter sublevel-set filtration that we defined in Section~\ref{sec:single_baseline}. It filters by distance to the nearest landfill, without accounting for the amount of waste at each landfill. 
As in our case study of public parks, we identify the components that arise from 0D PH (see Figure~\ref{fig:Landfill_Results}c), whose sizes are under some threshold.
These components are analogous to the truncated PH classes in our bifiltration.
Unlike the maximal components in our bifiltration (see Figure~\ref{fig:Landfill_Results}a), these components are numerous and cover most of Illinois.

\begin{figure}
    \centering
    \includegraphics[width=\linewidth]{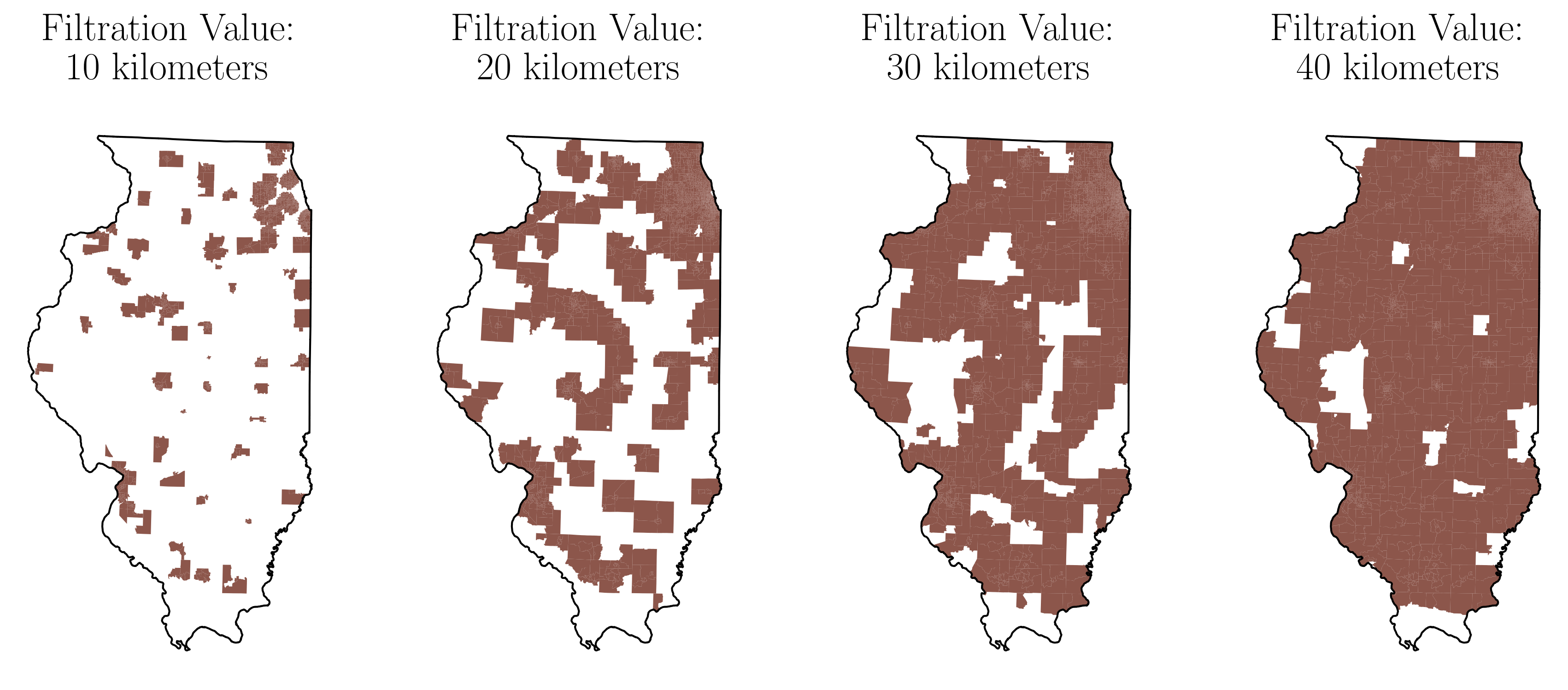}
    \caption{Four steps in the naive superlevel-set filtration (see Section~\ref{sec:single_baseline}) for landfills in Illinois.
    }
\label{fig:Landfills_Sublevel_Filtration}
\end{figure}


\subsection{Pubs in Chicago}\label{sec:pubs}

In our third case study, we examine 
bars and pubs
in Chicago. The geographic region $D$ is Chicago, the set $\mathcal{B}$ of geographic units is the set of census blocks, and the sites $\{M_i\}$ are pubs.

Unlike in our previous case studies, we examine bars and pubs in Chicago from two different perspectives: 
\begin{itemize}
    \item[(1)]{Resources: Establishments like bars and pubs provide drinks, food, and social spaces for for their patrons.}
    \item[(2)]{Nuisances: Bars and pubs can annoy local residents that are sensitive to noise, 
    want to minimize their exposure to drinking culture, or for other reasons.} 
\end{itemize}
For these two different perspectives, we construct a bifiltration for each, in which we use
different distance measures and scores to 
capture the impact of bars and pubs on people.


\subsubsection{Scores} 

We obtain our scores for pubs from Google Reviews \cite{GoogleAPI}. 
When we examine pubs as resources, a pub's score is its total stars from reviews.
When we examine pubs as nuisances, a pub's score is its total number of reviews.
A pub's total number of stars estimates its popularity.
Star ratings reflect the perceived quality of a pub as a resource, so a higher rating (i.e., a rating with more stars) indicates a higher perceived quality.
A pub's number of ratings reflects the number of people that visit it. We use this number of ratings a proxy for how many people visit a pub and hence for its sensory impact (and its potential as a nuisance) in its
neighborhood.
In Figure~\ref{fig:pubsgood_and_bad}, we show the pubs and their scores.

\begin{figure}
    \centering
    \includegraphics[width=0.48\linewidth]{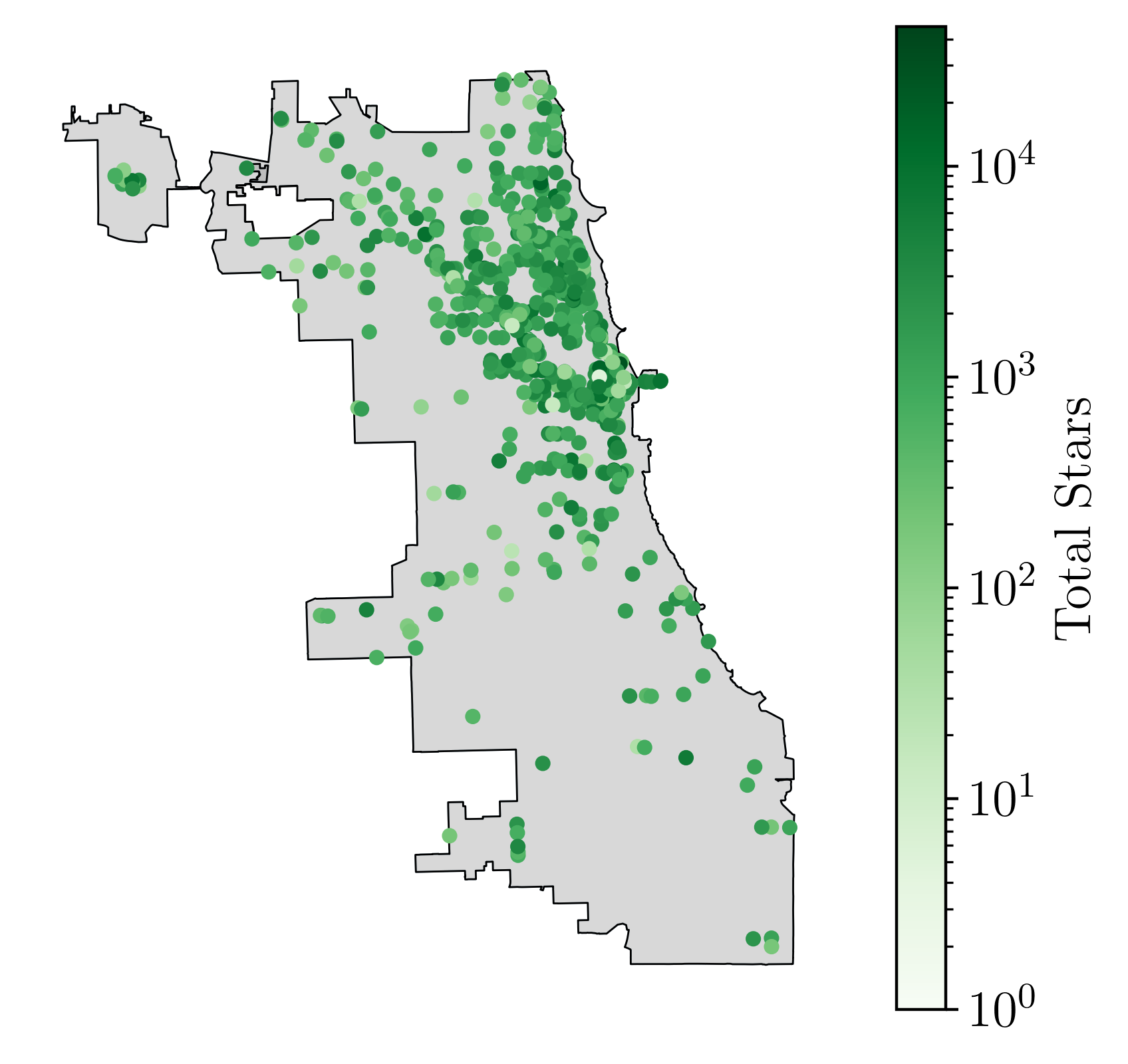}
    \includegraphics[width=0.48\linewidth]{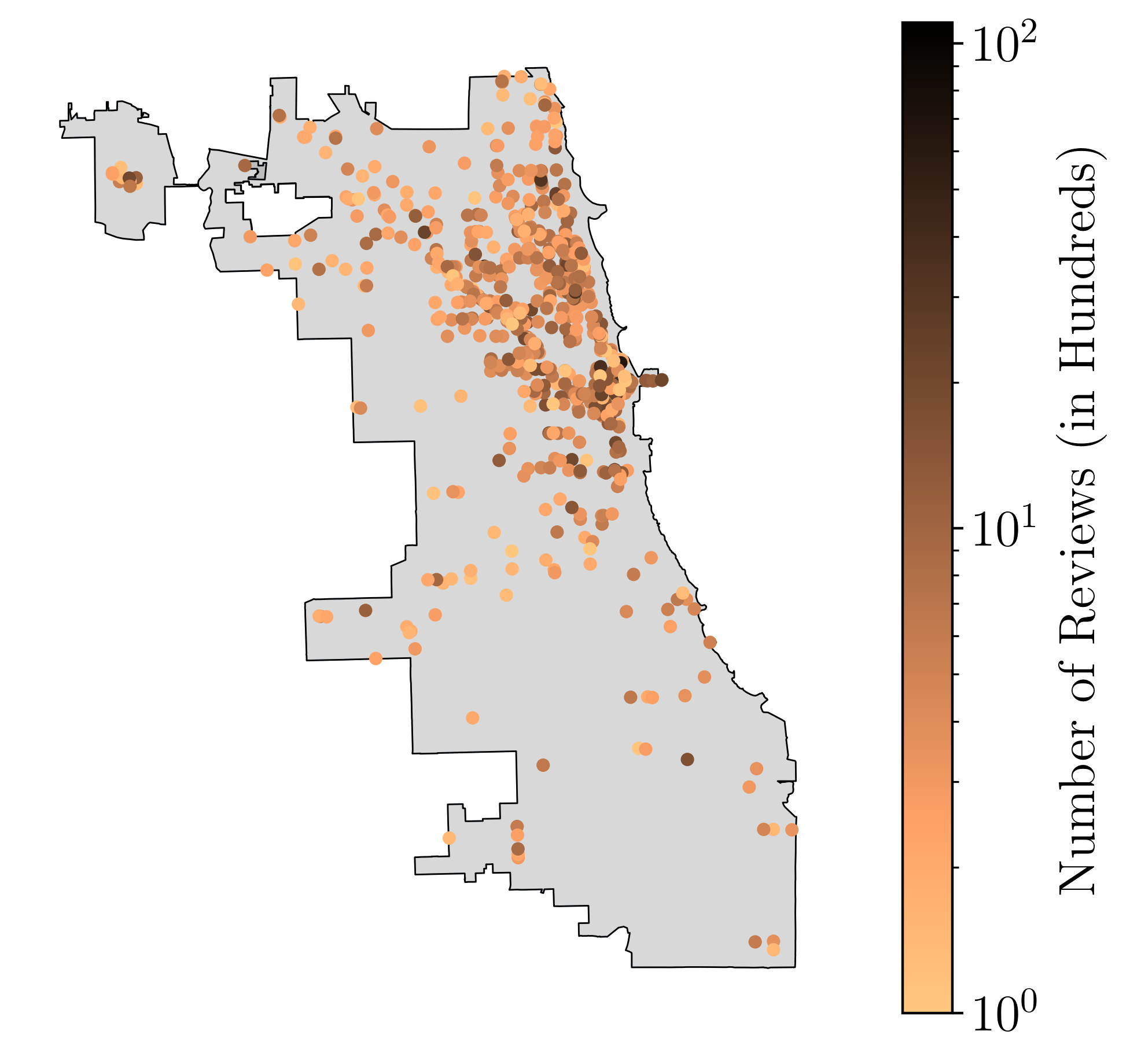}
    \caption{Bars and pubs in the city of Chicago. We color them by (left) the total stars in the reviews
     and (right) the total number of reviews. 
    Both color bars use logarithmic scales, but their value ranges are different.
    }
    \label{fig:pubsgood_and_bad}
\end{figure}


\subsubsection{Distances} 

As in our case study of parks, when treating pubs as a resource, we define distances from pubs to census blocks by considering accessibility, in which we use $L_1$ distance to approximate city-street distance.
As in our case study of landfills, when treating pubs as a nuisance, we define distance by considering physical proximity, in which we use Euclidean distance.


\subsubsection{Results}

We show the maximal components when treating pubs as resources in Figure~\ref{fig:Pubs_Results}a, and we show the maximal components when treating pubs as nuisances in Figure~\ref{fig:Pubs_Results}b.
When considering pubs as a resource, some of these maximal components are simply regions without any
pubs.
Counterintuitively, however, some of the maximal components with the longest lifetimes (i.e., the regions with poorest pub coverage) are near the high-density region of pubs in the northeast.
When we treat pubs as nuisances, one of the maximal components 
includes
a high-density region, as expected. 
However, it is not the component with the longest lifetime; several smaller maximal components have longer lifetimes, again highlighting the fact that there is more to nuisance exposure and resource coverage than geospatial density.
The maximal component with the longest lifetime includes Chicago O'Hare International Airport, which lies in the northwest corner of the city.
Many people travel through that airport, so it makes sense that nearby pubs have so many reviews.

\begin{figure}
    \centering
    \subfloat[Pubs as resources]{\includegraphics[width=0.45\linewidth]{PUBS_resource_Maxrep_ALL_lifetime_areaperc_0p2_0p9}}\label{fig:Pubs_Resource_Results}
    \hspace{5mm}
    \subfloat[Pubs as Nuisances]{\includegraphics[width=0.45\linewidth]{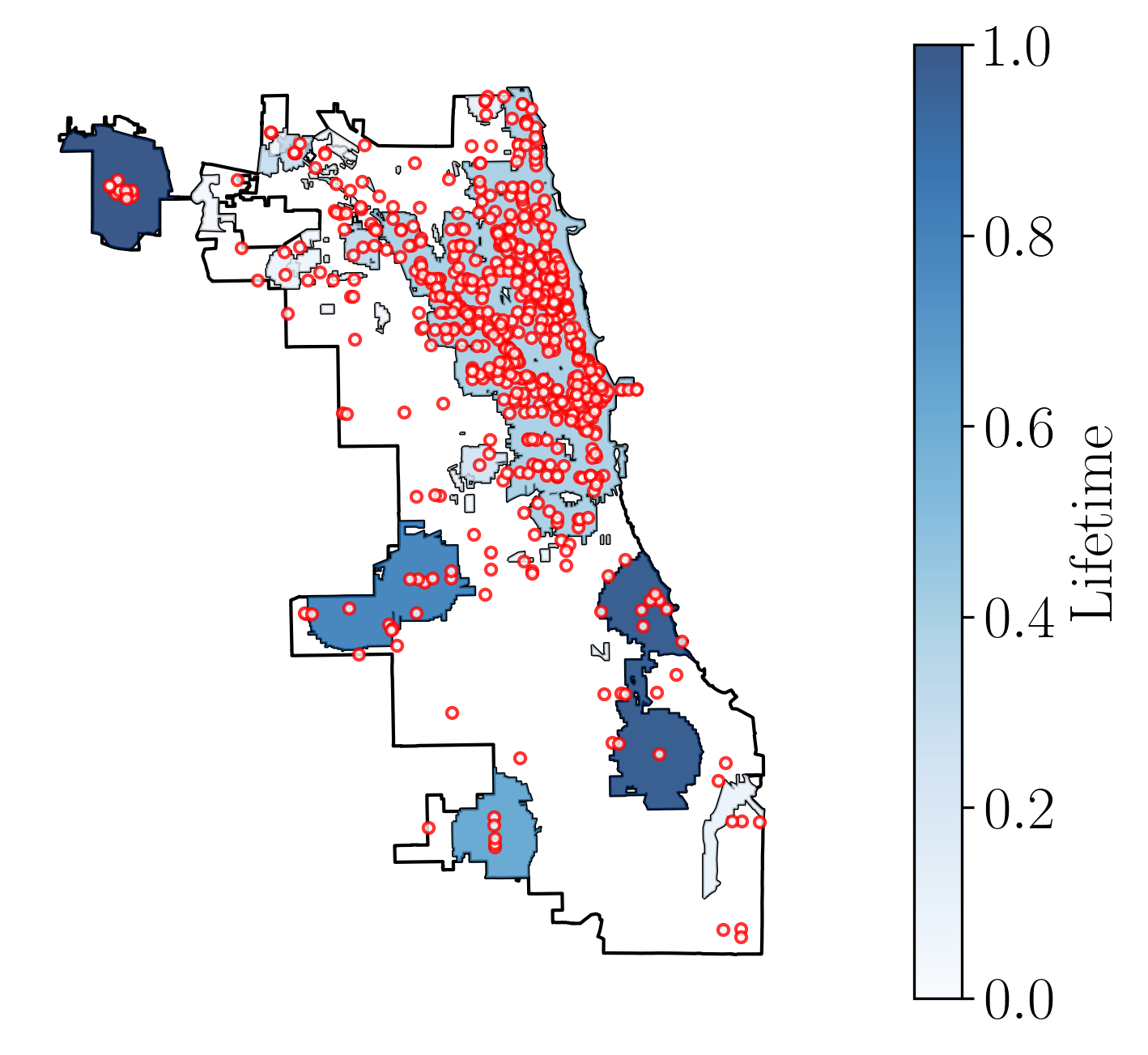}}\label{fig:Pubs_Nuisance_Results}
    \caption{Maximal components for our examinations of pubs as (a) resources and (b) nuisances. The red dots are the locations of pubs and bars. The components are colored by their lifetimes (see Section~\ref{sec:range}).    }
    \label{fig:Pubs_Results}
\end{figure}


\subsubsection{Comparison to Our One-Parameter PH Approach}

As a comparison, we show the first four steps of naive one-parameter filtrations for treating pubs as a resource and as a nuisance in Figure~\ref{fig:Pubs_Filtrations}. 
We show the truncated components in Figure~\ref{fig:pubs_baseline_components}.

\begin{figure}
    \centering
    \includegraphics[width=\linewidth]{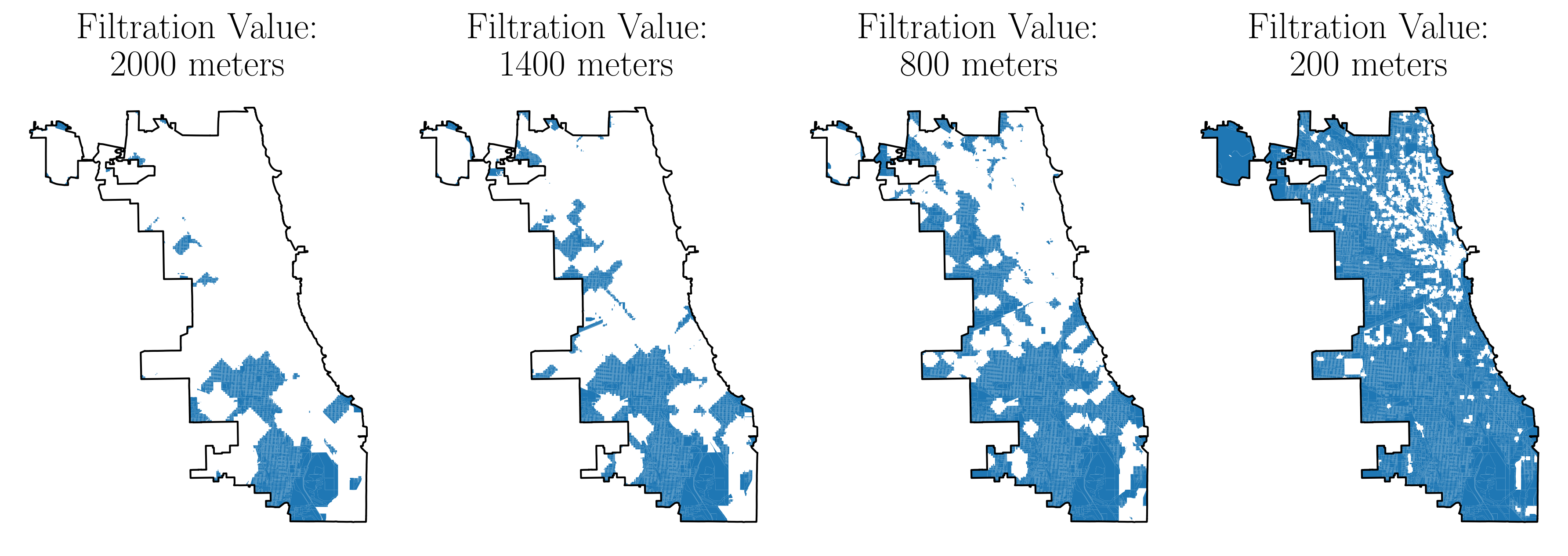}
    \includegraphics[width=\linewidth]{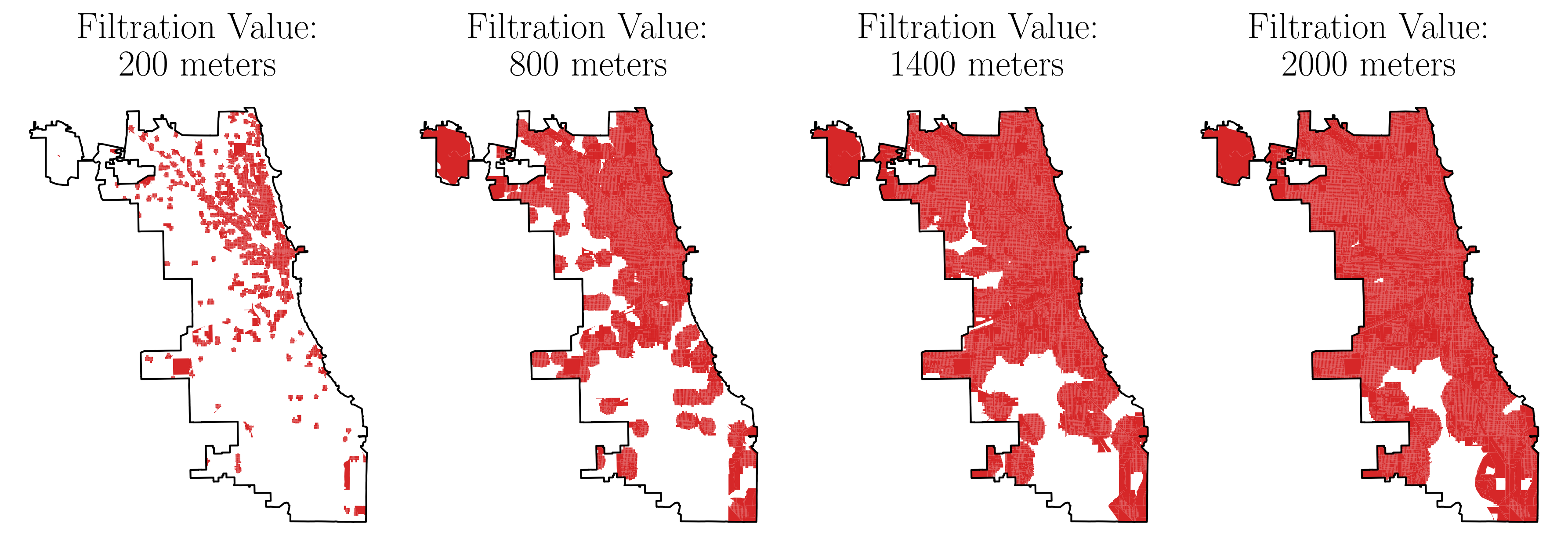}
    
    \caption{Four steps in the (top) superlevel-set filtration and (bottom) sublevel-set filtration (see Section~\ref{sec:single_baseline})
    for bars and pubs in Chicago.
    }
    \label{fig:Pubs_Filtrations}
\end{figure}

\begin{figure}
    \centering
    \subfloat[Pubs (resource)]{\includegraphics[width=0.45\linewidth]{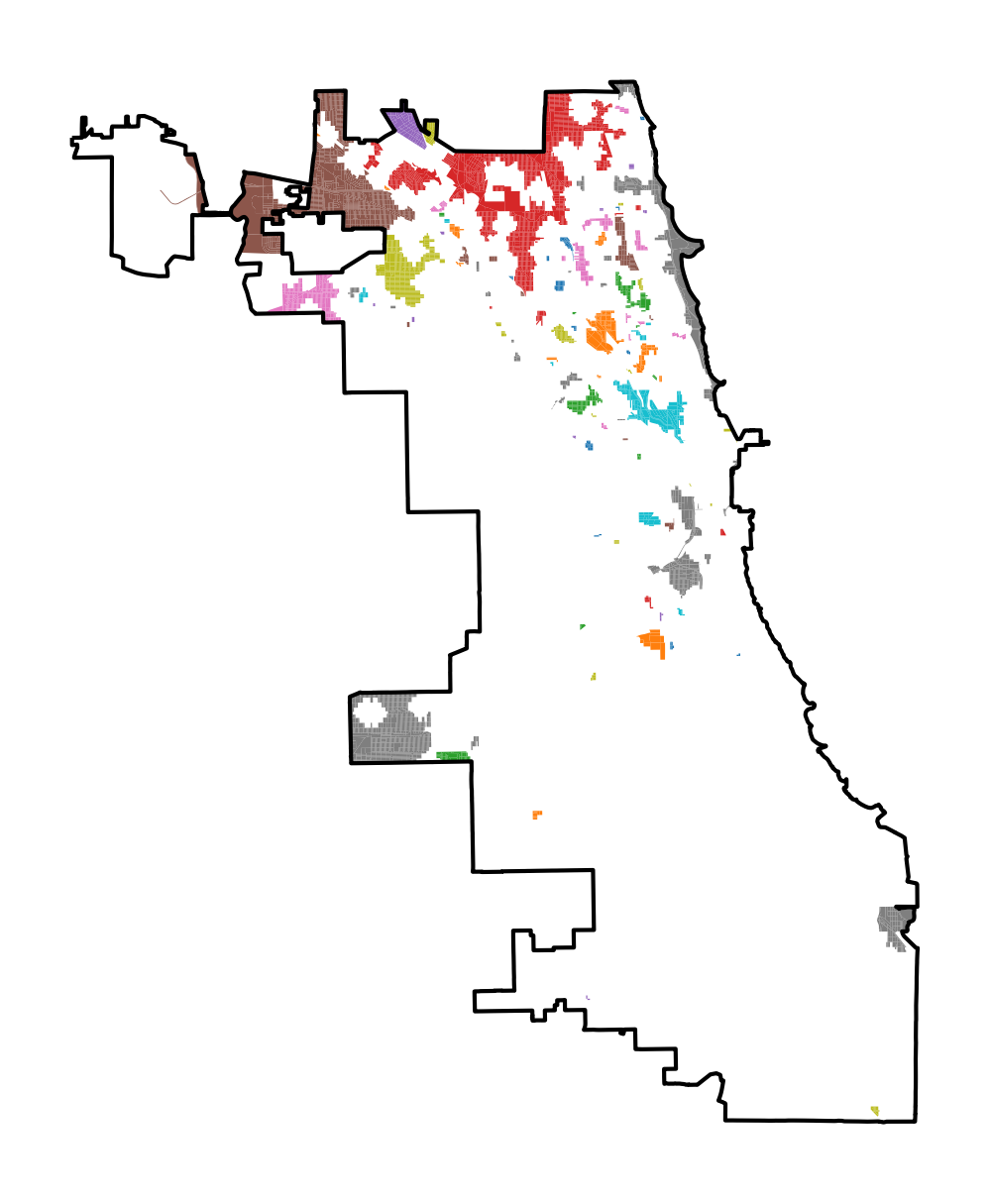} }
    \subfloat[Pubs (nuisance)]{\includegraphics[width=0.45\linewidth]{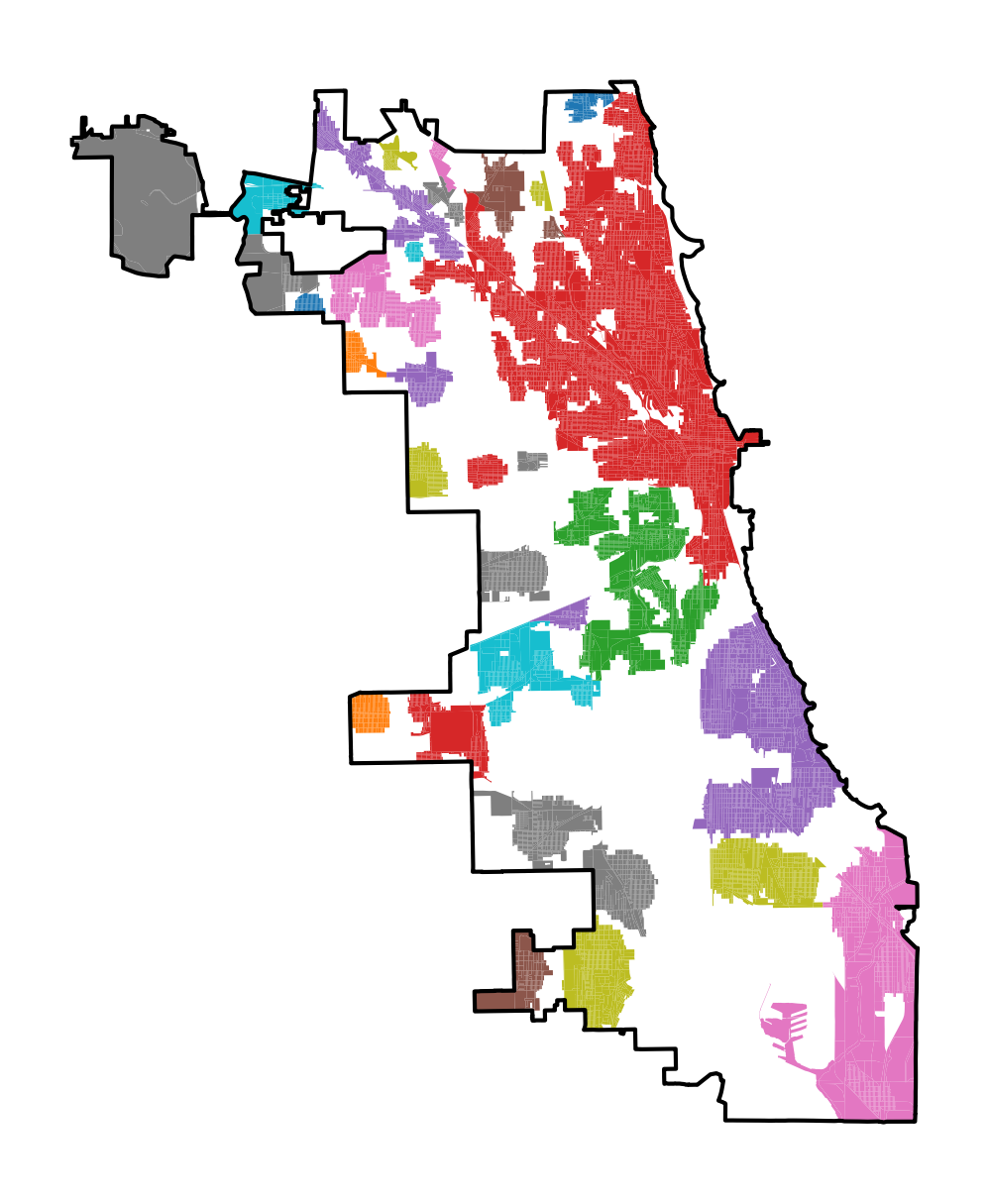}}
    \caption{Components in one-parameter filtrations of Chicago pubs as (a) resources (in which we use superlevel-set filtration) and (b) nuisances (in which we use sublevel-set filtration). We randomly assign colors to help distinguish between neighboring components.
    }
    \label{fig:pubs_baseline_components}
\end{figure}



\section{Conclusions and Discussion}\label{conc} 

To assess resource access and nuisance exposure, including for the evaluation of equity, one needs to analyze both geographic access/exposure and and resource/nuisance amount.
One way to do this is to employ approaches from topological data analysis (TDA).


\subsection{Summary}

We developed a multiparameter persistent-homology (PH) \linebreak framework to analyze access to resources and exposure to nuisances in a manner that accounts both for their locations and for their quality/severity.
Multiparameter PH
is notoriously challenging to interpret. 
We developed a method to capture information the multiparameter PH that is both computationally efficient and useful in practice. As case studies, we examined three different applications---public parks, landfills, and bars and pubs---and identified regions with the poorest access to resources or the most exposure to nuisances. We also compared the regions that we identified using our framework with regions from a one-parameter PH approach.

Incorporating both distance and quality/severity parameters in our \linebreak two-parameter PH framework allowed us to provide a systematic and nuanced evaluation of resource coverage and nuisance exposure.
In particular, we identified regions that have poor access to high-quality resources or overexposure to high-severity nuisances.
By contrast, when we employed one-parameter PH, we predominantly captured areas that are not very close to a resource or a nuisance; we did not account for the quality or severity of those nearby resources or nuisances, respectively.
However, some nearby resources have such low quality that easy access to them is little improvement 
on living far away from any resources.
Similarly, significant exposure to a low-severity nuisance is not much more harmful than little exposure to any nuisance.
Accordingly, employing a naive one-parameter PH approach yields a notably incomplete picture of resource access and nuisance exposure.


\subsection{Discussion}

Our PH framework requires many choices, which influence both our qualitative and quantitative results.
The ability to make such choices is both a benefit and a drawback of our approach. The flexibility of our approach is very beneficial, but how should one make these choices? One can explore alternative choices to those in our paper using code in our repository \cite{ourgithub}.

One key choice in our framework is the ability to employ any desired measure of resource quality or nuisance severity. We used total area for parks and two different ways of quantifying Google Reviews for bars and pubs, but there are numerous possible quality and severity measures.
For example, one can use Google Reviews of parks (we include this data in our GitHub repository \cite{ourgithub}) or create a score based on park attractions (e.g., playgrounds and athletic facilities), which is available at the Chicago Data Portal \cite{chicagodata}. 
It is convenient that one can explore the availability of resources using any desired quality score, but how does one compare or decide between qualitatively different results that arise from different choices?

Another choice is the size threshold to determine when an aggregated maximal component is ``too large.'' 
We examined (1) a threshold that is based on the number of census blocks in a component and (2) a threshold that is based on the area of a component's geographic region. The former is more computationally efficient, but the latter accounts for the fact that census blocks and tracts do not have homogeneous sizes.
Depending on the choice of threshold, 
one may not sufficiently capture large maximal components. 
For example, in Figures~\ref{fig:CompareCutoffs_Pubs}(a) and (b), we use thresholds of 10\% and 20\%, respectively, of Chicago's total area.
Most
of the maximal components are the same, but Figure~\ref{fig:CompareCutoffs_Pubs}b includes a large component in the southern half of Chicago that does not appear in Figure~\ref{fig:CompareCutoffs_Pubs}a. 
However, it is easy to identify such large regions by visual inspection, and we seek to find more subtle underresourced or overexposed regions through our analysis.

\begin{figure}
    \centering
    \subfloat[S=10\%]{\includegraphics[width=0.49\linewidth]{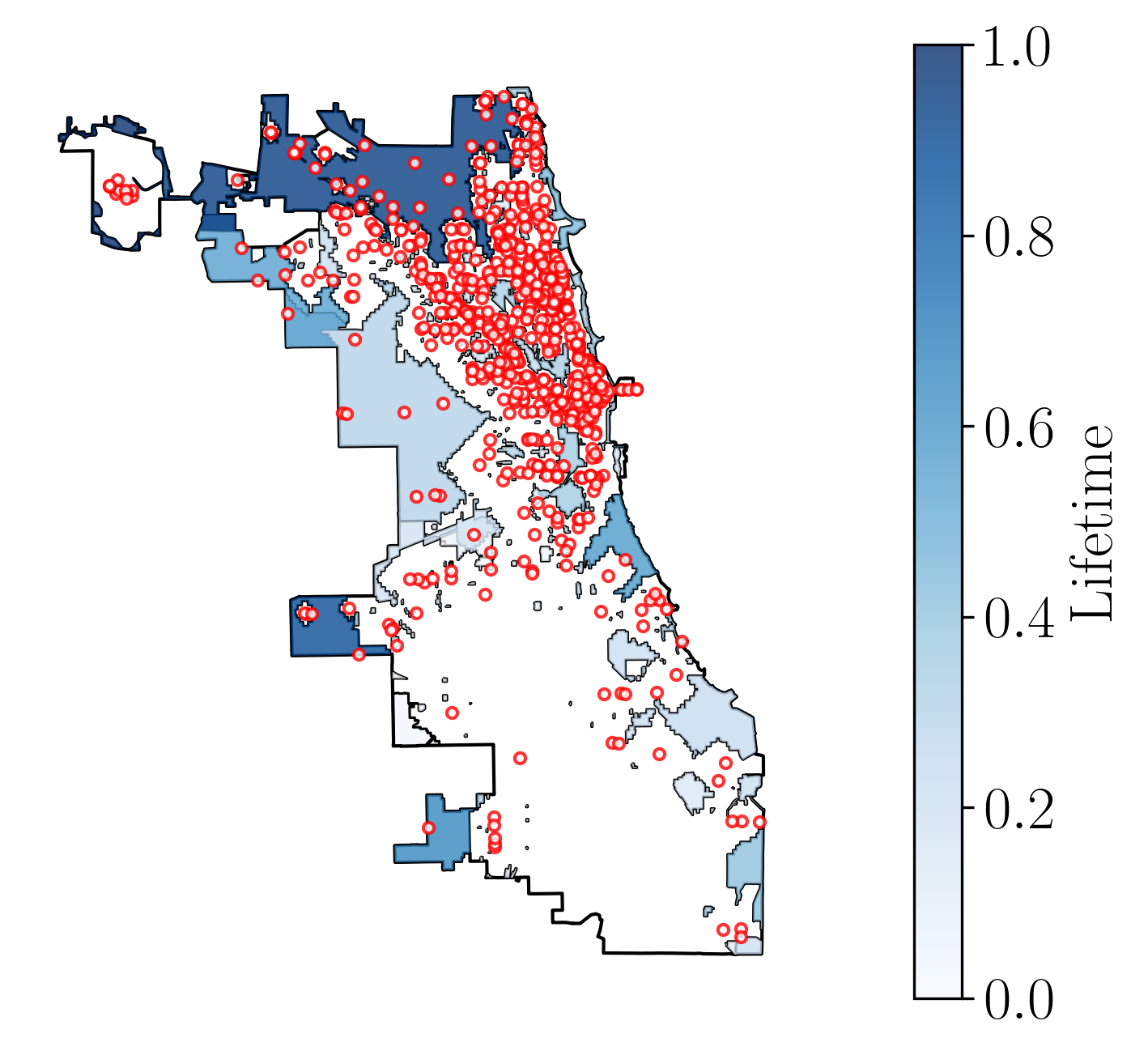}}\label{fig:CompareCutoffs_Pubs10}
    \subfloat[S=20\%]{\includegraphics[width=0.49\linewidth]{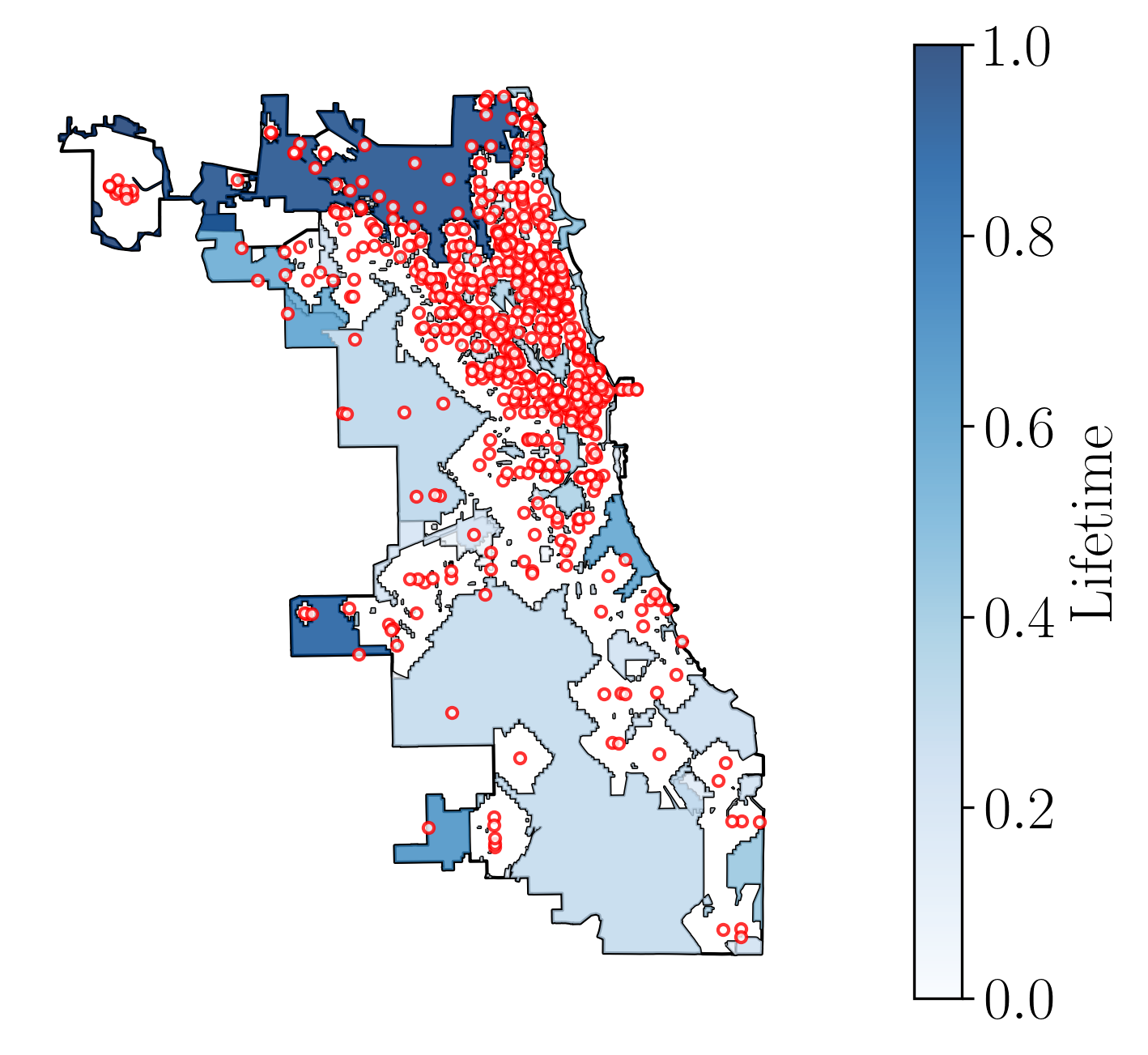}}\label{fig:CompareCutoffs_Pubs20}
    \caption{Maximal components of Chicago pubs as a resource when we use (a) a size threshold of $S = 10$\% and (b) a size threshold of $S = 20$\%.}
    \label{fig:CompareCutoffs_Pubs}
\end{figure}


\appendix

\section{Data Collection} \label{app:data}
We now give details about our data collection.


\subsection{Public Parks in Chicago}

\subsubsection{Site Data}

We obtained the list and locations of Chicago parks from the Chicago Data Portal \cite{chicagodata}.
This data set includes the polygons that encode the park geometries.


\subsubsection{Quality Score}

The data set that we retrieved from the Chicago Data Portal \cite{chicagodata} includes the acreage of each park.


\subsection{Landfills in Illinois} 


\subsubsection{Site Data}

We obtained the locations of landfills through the Landfill Methane Outreach Program (LMOP) of the United States Environmental Protection Agency \cite{landfilldata}.
This program provides data for each state; we use the data from the state of Illinois.
The data set has several duplicate entries (e.g., the data for an original landfill and an expansion of that landfill are listed as two separate data points), so we remove the duplicates.\footnote{The location and measurement of the total amount of waste in short tons is the same across duplicate entries, so it does not matter which entry one removes as the duplicate.}
This yielded a list of 90 landfills. 


\subsubsection{Quality Score} 

The landfill data set includes entries of ``waste in place'', which measures the weight of the waste in a landfill in units of short tons.
We remove any landfills that are missing this data.


\subsection{Bars and Pubs in Chicago} 


\subsubsection{Site Data} 

To obtain the geographic coordinates (latitude and longitude) and associated names of bars and pubs in Chicago, we followed the approach of Cocoran and Jones \cite{corcoran_topological_2023} and used OpenStreetMap (OSM) \cite{openstreetmap}.
We collected all locations with amenity tags of \texttt{amenity=bar} or \texttt{amenity=pub}, which resulted in an initial set of 949 bars and pubs.


\subsubsection{Quality Score} 

Collecting Google Reviews for the pubs required careful querying
of the Google Places application-programming-interface (API) \cite{GoogleAPI} because some bar and pub names are very common across many cities and states. Therefore, in addition to querying using bar and pub name, we also included the latitudinal and longitudinal coordinates that we obtained from OSM in our search results.
We also checked that each bar and pub is still open (by checking that its ``operational status'' is not listed as permanently closed or temporarily closed) and that it is in Chicago.
When this yielded a single pub of a given name, we checked that the
name and latitudinal and longitudinal coordinates from Google Maps approximately matches those from OSM. 
If the Euclidean distance between the 
two (latitude, longitude) coordinate pairs is at most
$1\times 10^{-3}$ and the names have a Jaro--Winkler similarity \cite{winkler1990string} of at least
$0.5$, then we kept the query result.
We chose these specific numerical thresholds
through trial and error. When a query yield two or more pubs with a given name, we selected the result that is closest geographically to the OSM coordinates and then checked the aforementioned conditions (that the pub is still open and that it has a similar name and location as in the OSM listing).
In both situations, when these measurements were not below the stated
thresholds, we checked manually using Google Maps to see if the bar had changed names or 
locations.
If so, we manually updated the pub data by adding the star rating and number of reviews; otherwise, we removed the bar from the data set. 
Ultimately, we manually updated the data for 19 bars. 
Our pruning procedure narrowed the 949 pubs from OSM to 734 pubs.
This difference of about 200 bars and pubs
is likely due to outdated data in OSM. 
We consolidated this data into a spreadsheet, and we include it in our repository \cite{ourgithub}. 

\section*{Acknowledgements}

We thank Serena Li, Matt Jones, Daryl DeFord, Gabrielle Gutierrez, Brian Brubach, and Amir Barghi for helpful comments and conversations.


\bibliographystyle{siamplain}
\bibliography{gis-tda-v04}

@article{barbosa2018,
title = "Human mobility: {M}odels and applications",
journal = "Physics Reports",
year = "2018",
author = "Hugo Barbosa and Marc Barthelemy and Gourab Ghoshal and Charlotte R. James and Maxime Lenormand and Thomas Louail and Ronaldo Menezes and José J. Ramasco and Filippo Simini and Marcello Tomasini",
}

@article{wu2026,
author = {Yulin Wu and Wenjia Zhang},
title = {Is a life circle a circle? {T}opological data analysis on activity space delineation},
journal = {Annals of the American Association of Geographers},
year = {2026},
volume = {advanced access},
doi = {10.1080/24694452.2025.2608173}
}

@article{schindler2022,
title = {How far do people travel to use urban green space? {A} comparison of three {E}uropean cities},
journal = {Applied Geography},
volume = {141},
eid = {102673},
year = {2022},
author = {Mirjam Schindler and Marion {Le Texier} and Geoffrey Caruso}
}

@article{bruno2024,
title = {A universal framework for inclusive 15-minute cities},
journal = {Nature Cities},
volume = {1},
pages = {631--641},
year = {2024},
author = {Bruno, Matteo and {Monteiro Melo}, Hygor Piaget and Campanelli, Bruno and Loreto, Vittorio}
}

@article{anderson2022,
title = {Topological data analysis and {UNICEF} {Multiple Indicator Cluster Surveys}},
journal = {Journal of Quantitative Economics},
volume = {20},
pages = {281--309},
year = {2022},
author = {Jun Ru Anderson and Fahrudin Memi\'{c} and Ismar Voli\'{c}}
}

@misc{sfil2025,
	title = {Access graph: {A} novel graph representation of public transport networks for accessibility analysis},
	url = {https://arxiv.org/abs/2507.08361},
	publisher = {arXiv},
	author = {Tina \v{S}filigoj and Aljo\v{s}a Peperko and Oded Cats},
	year = {2025},
	note = {arXiv:2507.08361}
}

@misc{gorjian2025,
	title = {Green schoolyard investments influence local-level economic and equity outcomes through spatial-statistical modeling and geospatial analysis in urban contexts},
	url = {https://arxiv.org/abs/2507.14232},
	publisher = {arXiv},
	author = {Mahshid Gorjian},
	year = {2025},
	note = {arXiv:2507.14232}
}

@article{skipper2025,
author = {Skipper, Daphne and Agarwala, Susama and Murrell, Joshua and Rosenberg, Chad and Wilson, Dan and Logan, Tom M.},
title = {Optimizing fair geographic access to polling},
journal = {Election Law Journal: Rules, Politics, and Policy},
doi = {10.1089/elj.2024.0060},
}

@article{horton2025b,
title = {A scalable optimization approach for equitable facility location: {M}ethodology and transportation applications},
journal = {Transportation Research Part B: Methodological},
volume = {201},
eid = {103319},
year = {2025},
author = {Drew Horton and Joshua Murrell and Daphne Skipper and Emily Speakman and Tom Logan}
}

@Article{horton2025,
  Title                    = {Hundreds of grocery outlets needed across the United States to achieve walkable cities},
  Author                   = {Drew Horton and Tom M. Logan and Emily Speakman and Daphne Skipper},
  Journal                  = {Nature Communications},
  Year                     = {2025},
  volume 		= {16},
  eid 			= {6051}
}

@incollection{feng2025,
  author = {Michelle Feng},
  title = {Interpreting topology in the context of social science},
  booktitle    = {Mathematical and Computational Methods for Complex Social Systems},
  editor       = {Heather Z. Brooks and Michelle Feng and Mason A. Porter and Alexandria Volkening},
  pages        = {141--163},
  publisher    = {American Mathematical Society},
  year         = 2025,
  volume = {80},
  series       = {Proceedings of Symposia in Applied Mathematics},
  address      = {Providence, RI, USA},
}

@book{dey2022,
  title={Computational Topology for Data Analysis},
  author={Dey, Tamal Krishna and Wang, Yusu},
  year={2022},
  publisher={Cambridge University Press},
  address = {Cambridge, UK}
}

@Article{otter2017,
  Title                    = {A roadmap for the computation of persistent homology},
  Author                   = {Otter, N. and Porter, M. A. and Tillmann, U. and Grindrod, P. and Harrington, H. A.},
  Journal                  = {European Physical Journal --- Data Science},
  Year                     = {2017},
  volume = {6},
  eid = {17}
}

@article{otter2019,
author = {Harrington, Heather A. and Otter, Nina and Schenck, Hal and Tillmann, Ulrike},
title = {Stratifying Multiparameter Persistent Homology},
journal = {SIAM Journal on Applied Algebra and Geometry},
volume = {3},
number = {3},
pages = {439--471},
year = {2019}
}

@article{Carlsson2009,
  author = {Carlsson, Gunnar and Zomorodian, Afra},
  title = {The theory of multidimensional persistence},
  journal = {Discrete \& Computational Geometry},
  volume = {42},
  number = {1},
  pages = {71--93},
  year = {2009}
}

@misc{Lesnick2015,
	title = {Interactive visualization of {2-D} persistence modules},
	url = {https://arxiv.org/abs/1512.00180},
	publisher = {arXiv},
	author = {Lesnick, Michael and Wright, Matthew},
	year = {2015},
	note = {arXiv:1512.00180}
}

@misc{Botnan2022,
	title = {An introduction to multiparameter persistence},
	url = {https://arxiv.org/abs/2203.14289},
	publisher = {arXiv},
	author = {Botnan, Magnus B. and Lesnick, Michael},
	year = {2022},
	note = {arXiv:2203.14289}
}

@software{openstreetmap,
  author = {{OpenStreetMap contributors}},
  title = {{O}pen{S}treet{M}ap},
  url = {https://www.openstreetmap.org},
  note = {Version 1.2.2. (Accessed 11 August 2023)},
  version = {1.2.2},
  year = {2021}
}

@software{rivet,
        author = {{The RIVET Developers}},
        title = {{RIVET}},
        url = {https://github.com/rivetTDA/rivet/},
        version = {1.1.0},
        note = {Version 1.1.0. (Accessed 23 December 2025)}
}

@book{third_place,
  title={The Great Good Place: Caf{\'e}s, Coffee Shops, Community Centers, Beauty Parlors, General Stores, Bars, Hangouts, and How They Get You Through the Day},
  author={Ray Oldenburg},
  isbn={9781557781109},
  lccn={89003235},
  year={1989},
  publisher={Paragon House Publishers},
  address = {St. Paul, MN, USA}
}

@article{feng2020,
  title = {Spatial applications of topological data analysis: {C}ities, snowflakes, random structures, and spiders spinning under the influence},
  author = {Feng, Michelle and Porter, Mason A.},
  journal = {Physical Review Research},
  volume = {2},
  issue = {3},
  eid = {033426},
  year = {2020}
}

@incollection{tda_spatial,
  author = {Michelle Feng and Abigail Hickok and Mason A. Porter},
  title = {Topological data analysis of spatial systems},
  booktitle    = {Higher-Order Systems},
  editor       = {Federico Battiston and Giovanni Petri},
  chapter      = 16,
  pages        = {389--399},
  publisher    = {Springer},
  year         = 2022,
  series       = {Understanding Complex Systems},
  address      = {Cham, Switzerland},
}

@article{feng2021,
  title={Persistent homology of geospatial data: {A} case study with voting},
  author={Michelle Feng and Mason A. Porter},
  journal={SIAM Review},
  volume = {63},
  number = {1},
  pages = {67--99},
  year={2021}
}

@article{covid,
    title = {Analysis of spatial and spatiotemporal anomalies using persistent homology: {C}ase studies with {COVID}-19 data},
    author = {A. Hickok and D. Needell and M. A Porter},
    journal = {SIAM Journal on Mathematics of Data Science},
    volume = {4},
    number = {3},
    pages = {1116--1144},
    year = {2022}
}

@article{corcoran_topological_2023,
	title = {Topological data analysis for geographical information science using persistent homology},
	volume = {37},
	number = {3},
	journal = {International Journal of Geographical Information Science},
	author = {Corcoran, Padraig and Jones, Christopher B.},
	year = {2023},
	pages = {712--745},
}

@article{carmody_topological_2021,
	title = {Topological analysis of traffic pace via persistent homology},
	volume = {2},
	issn = {2632-072X},
	number = {2},
	urldate = {2023-02-12},
	journal = {Journal of Physics: Complexity},
	author = {Carmody, Daniel R. and Sowers, Richard B.},
	year = {2021},
	eid = {025007}
}

@article{corcoran_persistent_2022,
	title = {A persistent homology model of street network connectivity},
	volume = {26},
	number = {1},
	journal = {Transactions in GIS},
	author = {Corcoran, Padraig and Jones, Christopher B.},
	year = {2022},
	pages = {155--181},
}

@article{duchin_homological_2020,
title = {The (homological) persistence of gerrymandering},
author = {Moon Duchin and Tom Needham and Thomas Weighill},
journal = {Foundations of Data Science},
volume = {4},
number = {4},
pages = {581--622},
year = {2022}
}

@article{clause_meta-diagrams_2023,
	title = {Meta-diagrams for 2-parameter persistence},
	volume = {74},
	journal = {Discrete \& Computational Geometry},
	author = {Clause, Nate and Dey, Tamal K. and M\'{e}moli, Facundo and Wang, Bei},
	year = {2025},
	pages = {872--989}
}

@book{Edelsbrunner_Harer_2009, 
 address={Providence, RI, USA}, 
 title={Computational Topology}, 
 publisher={American Mathematical Society}, 
 author={Edelsbrunner, Herbert and Harer, John}, 
 year={2009}
 }

@article{qin_spatial_2020,
	title = {Spatial accessibility analysis of parks with multiple entrances based on real-time travel: {The} case study in {Beijing}},
	volume = {12},
	number = {18},
	journal = {Sustainability},
	author = {Qin, Jiahui and Liu, Yusi and Yi, Disheng and Sun, Shuo and Zhang, Jing},
	year = {2020},
	eid = {7618}
}

@article{Hickok_Jarman_Johnson_Luo_Porter_2024, 
 title={Persistent homology for resource coverage: {A} case study of access to polling sites}, 
 volume={66}, 
 number={3}, 
 journal={SIAM Review}, 
 author={Hickok, Abigail and Jarman, Benjamin and Johnson, Michael and Luo, Jiajie and Porter, Mason A.}, 
 year={2024}, 
 pages={481--500}
 }

@article{vipond_multiparameter_2021,
	title = {Multiparameter persistent homology landscapes identify immune cell spatial patterns in tumors},
	volume = {118},
	number = {41},
	journal = {Proceedings of the National Academy of Sciences of the United States of America},
	author = {Vipond, Oliver and Bull, Joshua A. and Macklin, Philip S. and Tillmann, Ulrike and Pugh, Christopher W. and Byrne, Helen M. and Harrington, Heather A.},
	year = {2021},
	eid = {e2102166118},
}

@article{radke_spatial_2000,
	title = {Spatial {Decompositions}, {Modeling} and {Mapping} {Service} {Regions} to {Predict} {Access} to {Social} {Programs}},
	volume = {6},
	issn = {1082-4006},
	number = {2},
	journal = {Geographic Information Sciences},
	author = {Radke, John and Mu, Lan},
	year = {2000},
	pages = {105--112},
}

@article{luo_measures_2003,
	title = {Measures of {Spatial} {Accessibility} to {Healthcare} in a {GIS} {Environment}: {Synthesis} and a {Case} {Study} in {Chicago} {Region}},
	volume = {30},
	issn = {0265-8135},
	shorttitle = {Measures of {Spatial} {Accessibility} to {Healthcare} in a {GIS} {Environment}},
	number = {6},
	journal = {Environment and Planning. B, Planning \& Design},
	author = {Luo, Wei and Wang, Fahui},
	month = dec,
	year = {2003},
	pmid = {34188345},
	pmcid = {PMC8238135},
	pages = {865--884}
}

@article{xing_environmental_2020,
	title = {An environmental justice study on spatial access to parks for youth by using an improved {2SFCA} method in {Wuhan}, {China}},
	volume = {96},
	journal = {Cities},
	author = {Xing, Lijun and Liu, Yanfang and Wang, Baoshun and Wang, Yiheng and Liu, Haijiang},
	year = {2020},
	eid = {102405},
}

@article{dai_measuring_2022,
	title = {Measuring equality in access to urban parks: {A} big data analysis from {Chengdu}},
	volume = {10},
	doi = {https://doi.org/10.3389/fpubh.2022.1022666},
	journal = {Frontiers in Public Health},
	author = {Dai, Weiwei and Yuan, Suyang and Liu, Yangyang and Peng, Dan and Niu, Shaofei},
	year = {2022},
}

@Article{liu_2sfca_elderly_care,
AUTHOR = {Liu, Linggui and Lyu, Han and Zhao, Yi and Zhou, Dian},
TITLE = {An improved two-step floating catchment area ({2SFCA}) method for measuring spatial accessibility to elderly care facilities in {Xi'an, China}},
JOURNAL = {International Journal of Environmental Research and Public Health},
VOLUME = {19},
YEAR = {2022},
NUMBER = {18},
eid = {11465}
}

@Article{jin_huff-based_2sfca,
AUTHOR = {Jin, He and Lu, Yongmei},
TITLE = {Multi-mode {H}uff-based {2SFCA}: {E}xamining geographical accessibility to food outlets in {Austin, Texas}},
JOURNAL = {ISPRS International Journal of Geo-Information},
VOLUME = {11},
YEAR = {2022},
NUMBER = {11},
eid = {579},
}

@article{KaufmannVisputeKansal,
   author = "Talia  Kaufmann and Swapnil  Vispute and Mansi  Kansal and Daniel T. O'Brien and Tomer  Shekel and Evgeniy  Gabrilovich and Gregory A. Wellenius and Lewis Dijkstra and Paolo Veneri",
   title = "Comparing access to urban parks across six {OECD} countries",
   year = "2023",
   journal = {OECD Regional Development Papers},
   eid = "45",
   doi = {https://doi.org/10.1787/4d17ce4c-en}
}

@article{SCHIPPERIJN2017253,
title = {Access to parks and physical activity: {A}n eight country comparison},
journal = {Urban Forestry \& Urban Greening},
volume = {27},
pages = {253--263},
year = {2017},
author = {Jasper Schipperijn and Ester Cerin and Marc A. Adams and Rodrigo Reis and Graham Smith and Kelli Cain and Lars B. Christiansen and van Dyck, Delfien and Christopher Gidlow and Lawrence D. Frank and Josef Mit\'{a}\v{s} and Michael Pratt and Deborah Salvo and Grant Schofield and James F. Sallis}
}

@article{ZHANG2023104191,
title = {Accessibility in a multiple transport mode urban park based on the {``D-D"} model: {A} case study in {Park City, Chengdu}},
journal = {Cities},
volume = {134},
eid = {104191},
year = {2023},
author = {Ziyang Zhang and Guoqiang Ma and Xiang Lin and Haoyu Dai}
}

@article{sturm_proximity_park_2014,
	title = {Proximity to urban parks and mental health},
	volume = {17},
	number = {1},
	journal = {The Journal of Mental Health Policy and Economics},
	author = {Sturm, Roland and Cohen, Deborah},
	year = {2014},
	pages = {19--24},
}

@article{kaczynski_proximity_park_2014,
	title = {Are park proximity and park features related to park use and park-based physical activity among adults? {Variations} by multiple socio-demographic characteristics},
	volume = {11},
	shorttitle = {Are park proximity and park features related to park use and park-based physical activity among adults?},
	number = {1},
	journal = {International Journal of Behavioral Nutrition and Physical Activity},
	author = {Kaczynski, Andrew T. and Besenyi, Gina M. and Stanis, Sonja A. Wilhelm and Koohsari, Mohammad Javad and Oestman, Katherine B. and Bergstrom, Ryan and Potwarka, Luke R. and Reis, Rodrigo S.},
	year = {2014},
	eid = {146},
}

@article{vrijheid_landfill_2000,
	title = {Health effects of residence near hazardous waste landfill sites: {A} review of epidemiologic literature},
	volume = {108, Supplement 1},
	shorttitle = {Health effects of residence near hazardous waste landfill sites},
	number = {Supplement 1},
	journal = {Environmental Health Perspectives},
	author = {Vrijheid, M.},
	year = {2000},
	pages = {101--112},
}

@article{njoku_landfill_2019,
	title = {Health and environmental risks of residents living close to a landfill: {A} case study of {Thohoyandou} landfill, {Limpopo} {Province}, {South} {Africa}},
	volume = {16},
	number = {12},
	journal = {International Journal of Environmental Research and Public Health},
	author = {Njoku, Prince O. and Edokpayi, Joshua N. and Odiyo, John O.},
	year = {2019},
	eid = {2125},
}

@article{auchincloss_alcohol_2022,
	title = {Alcohol outlets and alcohol consumption in changing environments: {P}revalence and changes over time},
	volume = {17},
	issn = {1747-597X},
	shorttitle = {Alcohol outlets and alcohol consumption in changing environments},
	number = {1},
	urldate = {2023-08-14},
	journal = {Substance Abuse Treatment, Prevention, and Policy},
	author = {Auchincloss, Amy H. and Niamatullah, Saima and Adams, Maura and Melly, Steven J. and Li, Jingjing and Lazo, Mariana},
	month = feb,
	year = {2022},
	eid = {7},
}

@article{halonen_alcohol_living_2013,
	title = {Living in proximity of a bar and risky alcohol behaviours: {A} longitudinal study},
	volume = {108},
	number = {2},
	urldate = {2023-08-14},
	journal = {Addiction},
	author = {Halonen, Jaana I. and Kivim\"{a}ki, Mika and Virtanen, Marianna and Pentti, Jaana and Subramanian, S. V. and Kawachi, Ichiro and Vahtera, Jussi},
	year = {2013},
	pages = {320--328},
}

@article{seid_alcohol_2018,
	title = {Is proximity to alcohol outlets associated with alcohol consumption and alcohol-related harm in {Denmark}?},
	volume = {35},
	number = {2},
	journal = {Nordic Studies on Alcohol and Drugs},
	author = {Seid, Abdu K. and Berg-Beckhoff, Gabriele and Stock, Christiane and Bloomfield, Kim},
	year = {2018},
	pages = {118--130},
}

@article{winkler1990string,
  title={String comparator metrics and enhanced decision rules in the {Fellegi--Sunter} model of record linkage},
  author={Winkler, William E.},
  year={1990},
note = {U.S. Bureau of the Census, Statistics Research Division. Available at \url{https://files.eric.ed.gov/fulltext/ED325505.pdf}. (Accessed 23 December 2025)}
}

@website{chicagodata,
  author = {{C}hicago {P}ark {D}istrict},
  title = {{C}hicago {D}ata {Portal}},
  url = {https://data.cityofchicago.org/},
  note = {(Accessed 30 January 2023)}
}

@website{landfilldata,
  author = {United States Environmental Protection Agency},
  title = {{L}andfill {M}ethane {O}utreach {P}rogram ({LMOP})},
  url = {https://www.epa.gov/lmop},
  note = {(Accessed 20 July 2023)},
}

@article{Scoccola2023,
    year = {2023},
    publisher = {The Open Journal},
    volume = {8},
    number = {83},
    eid = {5022},
    author = {Luis Scoccola and Alexander Rolle},
    title = {Persistable: {P}ersistent and stable clustering},
    journal = {Journal of Open Source Software}
}

@article{Rolle2020,
    title = {Stable and consistent density-based clustering via multiparameter persistence},
	journal = {The Journal of Machine Learning Research},
	volume = {25},
	number = {1},
	eid = {258},
	author = {Alexander Rolle and Luis Scoccola},
	year = {2024}
}

@misc{oneil_evaluating_2024,
	title = {Evaluating cooling center coverage using persistent homology of a filtered witness complex},
	url = {http://arxiv.org/abs/2410.09067},
	publisher = {arXiv},
	author = {O'Neil, Erin and Tymochko, Sarah},
	year = {2024},
	note = {arXiv:2410.09067},
}

@misc{friesen_understanding_2024,
	title = {Understanding {U}.{S}. racial segregation through persistent homology},
	url = {http://arxiv.org/abs/2410.10886},
	publisher = {arXiv},
	author = {Friesen, Ori and Ziegelmeier, Lori},
	year = {2024},
	note = {arXiv:2410.10886}
}

@misc{kauba_demographics_2024,
	title = {Topological analysis of {U}.{S}. city demographics},
	url = {http://arxiv.org/abs/2410.10886},
	publisher = {arXiv},
	author = {Jakini A. Kauba and Thomas Weighill},
	year = {2023},
	note = {arXiv:2310.08334},
}

@article{moi_urban_topology_2024,
  title = {Urban topology and dynamics can assess the importance of green areas},
  author = {Moi, Jacopo and Chiesi, Leonardo and Chirici, Gherardo and Francini, Saverio and Borghi, Costanza and Costa, Paolo and Galmarini, Bianca and Caldarelli, Guido},
  journal = {Physical Review E},
  volume = {110},
  issue = {6},
  eid = {064128},
  year = {2024},
}

@misc{beyond_proximity,
	title = {Cities beyond proximity},
	url = {https://arxiv.org/abs/2411.12335},
	publisher = {arXiv},
	author = {Dan Hill and Matteo Bruno and Hygor Piaget Monteiro Melo and Yuichiro Takeuchi and Vittorio Loreto},
	year = {2024},
	note = {arXiv:2411.12335},
}

@misc{florence_walk,
	title = {Quantifying walkable accessibility to urban services: {A}n application to {F}lorence, {I}taly},
	url = {https://arxiv.org/abs/2504.12934},
	publisher = {arXiv},
	author = {Leonardo Boncinelli and Stefania Miricola and Eugenio Vicario},
	year = {2025},
	note = {arXiv:2504.12934},
}

@article{Bjerkevik_2025, 
 title={Stabilizing decomposition of multiparameter persistence modules}, 
 journal={Foundations of Computational Mathematics}, 
 author={Bjerkevik, H. B.}, 
 year={2025}, 
 doi = {https://doi.org/10.1007/s10208-025-09695-w}
 }

@article{Chung_Day_Hu_2022, 
 title={A multi-parameter persistence framework for mathematical morphology}, 
 volume={12}, 
 number={1}, 
 journal={Scientific Reports}, 
 author={Chung, Yu-Min and Day, Sarah and Hu, Chuan-Shen}, 
 year={2022}, 
 eid={6427}
 }

@inproceedings{Loiseaux2023, 
 author = {Loiseaux, David and Carri\`{e}re, Mathieu and Blumberg, Andrew J.}, 
 title = {A framework for fast and stable representations of multiparameter persistent homology decompositions}, 
 year = {2023}, 
 publisher = {Curran Associates Inc.}, 
 address = {Red Hook, NY, USA}, 
 booktitle = {Proceedings of the 37th International Conference on Neural Information Processing Systems}, 
 eid = {1553}, 
 location = {New Orleans, LA, USA}, 
 series = {NIPS '23}
  }

@article{Lesnick2022,
author = {Lesnick, Michael and Wright, Matthew},
title = {Computing Minimal Presentations and Bigraded Betti Numbers of 2-Parameter Persistent Homology},
journal = {SIAM Journal on Applied Algebra and Geometry},
volume = {6},
number = {2},
pages = {267--298},
year = {2022}
}

@misc{arulandu2025grapevine,
	title = {Through the grapevine: {V}ineyard distance as a measure of topological dissimilarity},
	url = {https://arxiv.org/abs/2510.24472},
	publisher = {arXiv},
	author = {{Arulandu}, Alvan and {Gottschalk}, Daniel and {Payne}, Thomas and {Richardson}, Alexander and {Weighill}, Thomas},
	year = {2025},
	note = {arXiv:2510.24472}
}

@misc{fairchild2025mortality,
	title = {Topological data analysis of mortality patterns during the {COVID-19} pandemic},
	url = {https://arxiv.org/abs/2510.15066},
	publisher = {arXiv},
	author = {{Fairchild}, Megan and {Lemoine}, Matthew},
	year = {2025},
	note = {arXiv:2510.15066}
}

@article{SCHERTZ2025102624,
title = {Identifying qualities and amenities associated with subjective cognitive restoration and improved affect in urban parks},
journal = {Journal of Environmental Psychology},
volume = {104},
eid = {102624},
year = {2025},
author = {Kathryn E. Schertz and Kimberly L. Meidenbauer and Tiara R. Freeman and Isabella M. Santiago and Elizabeth A. Janey and Kathryn Gehrke and Anya L. Samtani and Andrew J. Stier and Marc G. Berman}
}

@misc{weng2025beyonddistance,
	title = {Beyond distance: {M}obility neural embeddings reveal visible and invisible barriers in urban space},
	url = {https://arxiv.org/abs/2506.24061},
	publisher = {arXiv},
	author = {{Weng}, Guangyuan and {Kim}, Minsuk and {Ahn}, Yong-Yeol and {Moro}, Esteban},
	year = {2025},
	note = {arXiv:2506.24061}
}

@misc{rodriguez2025quantifying,
	title = {Quantifying displacement: {A} gentrification's consequence via persistent homology},
	url = {https://arxiv.org/abs/2512.10753},
	publisher = {arXiv},
	author = {{Rodr\'iguez V\'azquez}, Rita and {Cuerno}, Manuel}}

@misc{cruz-gonzalez2025sexual_healthcare,
	title = {Defunding sexual healthcare: {A} topological investigation of resource accessibility},
	url = {https://arxiv.org/abs/2512.12011},
	publisher = {arXiv},
	author = {Denise Gonzalez-Cruz and Genesis Encarnacion and Kaili Martinez-Beasley and Robin Wilson and Nicholas Arosemena and Atilio Barreda and Omayra Ortega and Daniel A. Cruz},
	year = {2025},
	note = {arXiv:2512.12011}
}

@misc{wu2026ecosystemservicedemand,
	title = {Ecosystem service demand relationship and trade-off patterns in urban parks across {C}hina},
	url = {https://arxiv.org/abs/2602.11442},
	publisher = {arXiv},
	author = {Shuyao Wu and Delong Li and Zhonghao Zhang},
	year = {2026},
	note = {arXiv:2602.11442}
}

@inproceedings{singh2007topological,
booktitle = {Eurographics Symposium on Point-Based Graphics},
editor = {M. Botsch and R. Pajarola and B. Chen and M. Zwicker},
title = {Topological Methods for the Analysis of High Dimensional Data Sets and 3D Object Recognition},
author = {Singh, Gurjeet and M{\'e}moli, Facundo and Carlsson, Gunnar},
year = {2007},
publisher = {The Eurographics Association},
ISSN = {1811-7813},
ISBN = {978-3-905673-51-7},
DOI = {10.2312/SPBG/SPBG07/091-100}
}

@article{moreno2021introducing,
  title={Introducing the “15-Minute City”: Sustainability, resilience and place identity in future post-pandemic cities},
  author={Moreno, Carlos and Allam, Zaheer and Chabaud, Didier and Gall, Catherine and Pratlong, Florent},
  journal={Smart cities},
  volume={4},
  number={1},
  pages={93--111},
  year={2021},
  publisher={Multidisciplinary Digital Publishing Institute}
}

@misc{GoogleAPI,
	title = {Google Places API},
	url = {https://developers.google.com/maps/documentation/places/web-service/place-details},
	author = {Google Developers},
	note = {Accessed January 2023.}
}

@misc{GoogleMapsAPI,
	title = {Google Maps API},
	url = {https://developers.google.com/maps/documentation/javascript},
	author = {Google Developers},
}

@article{multipers,
  title = {Multipers: {{Multiparameter Persistence}} for {{Machine Learning}}},
  shorttitle = {Multipers},
  author = {Loiseaux, David and Schreiber, Hannah},
  year = {2024},
  month = nov,
  journal = {Journal of Open Source Software},
  volume = {9},
  number = {103},
  pages = {6773},
  issn = {2475-9066},
  doi = {10.21105/joss.06773},
  langid = {english},
}

@article{loiseaux2025multi,
  title={Multi-parameter Module Approximation: an efficient and interpretable invariant for multi-parameter persistence modules with guarantees},
  author={Loiseaux, David and Carriere, Mathieu and Blumberg, Andrew J},
  journal={Journal of Applied and Computational Topology},
  volume={9},
  number={4},
  pages={1--60},
  year={2025},
  publisher={Springer}
}

@Misc{ourgithub,
  author =	 {Sarah Tymochko and Gillian Grindstaff and Abigail Hickok and Jerry Luo and Mason Porter},
  title =	 {\texttt{multiparam-resource-coverage} {G}ithub {R}epository},
  url = {https://github.com/sarahtymochko/multiparam-resource-coverage}
}

@article{vipond_landscapes_2020,
  title   = {Multiparameter persistence landscapes},
  author  = {Vipond, Oliver},
  journal = {Journal of Machine Learning Research},
  volume  = {21},
  number  = {61},
  pages   = {1--38},
  year    = {2020}
}


\end{document}